\documentclass[a4paper,twoside,11pt]{article}
\usepackage{RR}

\usepackage{hyperref}
\usepackage{figlatex}
\graphicspath{{fig/}}
\usepackage[T1]{fontenc}
\usepackage{amsfonts,amsmath,amssymb,amstext,latexsym,stmaryrd}
\usepackage{algorithm2e_extension}
\usepackage{amsthm}

\newtheorem{proposition}{Proposition}

\newtheorem{theorem}{Theorem}
\newtheorem{definition}{Definition}

\newtheorem{remark}{Remark}
\newtheorem{example}{Example}
\usepackage{algorithmic, algorithm}

\let\leq=\leqslant
\let\geq=\geqslant

\newcommand{\bydef}{\stackrel{\rm{def}}{=}}

\begin{document}
\RRNo{6653}
\RRdate{May 2008}
\RRversion{2\thanks{A study of allocation games has been included.}}
\RRdater{September 2009} 

\RRauthor{Coucheney, Pierre and Touati, Corinne and Gaujal, Bruno}

\authorhead{P. Coucheney, C. Touati, B. Gaujal}

\RRtitle{Un algorithme distribu\'e pour une association
  utilisateur-r\'eseau efficace et \'equitable dans les r\'eseaux sans fils
  multi-technologiques }

\RRetitle{A Distributed Algorithm for Fair and Efficient User-Network
  Association in Multi-Technology Wireless Networks}

\titlehead{Mobile Centric Network Association Algorithm}

\RRresume{Les \'equipements mobiles r\'ecents (tels que d\'efinis dans la
  norme IEEE 802.21) permettent aux usagers de basculer d'une
  technologie \`a l'autre (ce que l'on nomme ``handover
  vertical''). Plus de souplesse est autoris\'ee dans l'allocation des
  ressources et, par cons\'equent, cela augmente potentiellement les
  d\'ebits allou\'es aux usagers.

  Dans cet article, nous concevons un algorithme distribu\'e qui proc\`ede
  par t\^atonnement pour obtenir une association utilisateur-r\'eseau
  efficace et \'equitable, afin d'exploiter les b\'en\'efices du ``handover
  vertical''. D'une part, les mobiles mettent \`a jour pas \`a pas la
  proportion de paquets de donn\'ees qu'ils envoient sur chaque r\'eseau \`a
  partir d'une valeur transmise par la station de base. D'autre part,
  les stations de base calculent et envoient cette valeur aux
  mobiles. Cette valeur, appel\'ee ``repercussion utility '' repr\'esente
  l'impact que chaque mobile a sur le d\'ebit global du r\'eseau.

  Cette fonction d'utilit\'e est \`a rapprocher de l'id\'ee du co\^ut marginal
  dans la litt\'erature sur la tarification. Aussi bien l'algorithme de
  la station de base que celui du mobile sont suffisamment simples
  pour \^etre impl\'ement\'es dans les \'equipements standards actuels.

  \`A partir de m\'ethodes des jeux \'evolutionnaires, des jeux de
  potentiel, de la dynamique de r\'eplication, et des approximations
  stochastiques, nous montrons de mani\`ere analytique la convergence de
  l'algorithme vers une solution efficace et \'equitable en terme de
  d\'ebit. De plus, nous montrons qu'une fois l'\'equilibre atteint,
  chaque utilisateur est connect\'e \`a un unique r\'eseau ce qui permet de
  supprimer le co\^ut du ``handover vertical''.

  Plusieurs heuristiques reposant sur cet algorithme sont propos\'ees
  afin d'obtenir une convergence rapide. En effet, pour des raisons
  d'ordre pratique, le nombre d'it\'erations doit demeurer de l'ordre de
  quelques dizaines. Nous comparons alors la qualit\'e des solutions
  fournies dans divers sc\'enarios.  }

\RRabstract{ Recent mobile equipment (as well as the norm IEEE 802.21)
  now offers the possibility for users to switch from one technology
  to another (vertical handover). This allows flexibility in resource
  assignments and, consequently, increases the potential throughput
  allocated to each user.

  In this paper, we design a fully distributed algorithm based on
  trial and error mechanisms that exploits the benefits of vertical
  handover to find  fair and efficient assignment schemes.  On the
  one hand, mobiles gradually update the fraction of data packets they
  send to each network based on a value called \emph{repercussion
    utility} they receive from the stations. On the other hand,
  network stations compute and send repercussion utilities to each mobile that
  represent the impact each mobile has on the cell throughput.

  This repercussion utility function is closely related to the concept
  of marginal cost in the pricing literature. Both the station and the
  mobile algorithms are simple enough to be implemented in current
  standard equipment.

  Based on tools from evolutionary games, potential games, replicator
  dynamics and stochastic approximations, we analytically show the
  convergence of the algorithm to solutions that are efficient and
  fair in terms of throughput. Moreover, we show that after
  convergence, each user is connected to a single network cell which
  avoids costly repeated vertical handovers.

  Several simple heuristics based on this algorithm are proposed to
  achieve fast convergence.  Indeed, for implementation purposes, the
  number of iterations should remain in the order of a few tens. We
  finally provide extensive simulation of the algorithm in several
  scenarios.

 }

 \RRkeyword{ Distributed Algorithms, Hybrid Wireless Networks,
   Evolutionary Games, Potential Games, Replicator Dynamics, Vertical
   Handover, Fairness, Stochastic Approximation.}

 \RRmotcle{Algorithmes distribu\'es, r\'eseaux sans-fils h\'et\'erog\`enes,
   interconnection de r\'eseau, th\'eorie des jeux \'evolutionnaires, jeux
   de potentiels, dynamique de r\'eplication, handover vertical,
   \'equit\'e, approximation stochastique.}

\RRprojet{MESCAL}
\URRhoneAlpes
\RRtheme{\THNum}
\makeRR

\section{Introduction}

The overall wireless market is expected to be served by six or more
major technologies (GSM, UMTS, HSDPA, WiFi, WiMAX, LTE). Each
technology has its own advantages and disadvantages and none of them
is expected to eliminate the rest.  Moreover, radio access equipment
is becoming more and more multi-standard, offering the possibility of
connecting through two or more technologies concurrently, using the
norm IEEE 802.21.  Switching between networks using different
technology is referred to as \emph{vertical handover}. This is
currently done in UMA, for instance, which gives an absolute priority
to WiFi over UMTS whenever a WiFi connection is available.  In this
paper, in contrast, we address the problem of computing an efficient
association by providing a distributed algorithm that can be fair to
all users or efficient in terms of overall throughput.
Here are  the theoretical contributions of the paper.\\
- First, we propose a distributed algorithm with guaranteed
convergence to a non-cooperative equilibrium. This algorithm is based
on an iterative mechanism: at each time epoch the mobile nodes adapt
the proportion of the traffic they send on each network, based on some
values (caled repercussion utilities in the following)  they receive from the network. This work is in line with some
recent work on learning of Nash equilibria (see, for instance,
\cite{barth08}
\cite{sastry94}).\\
- Second, based on tools from potential games, we show that, by
appropriately setting up the repercussion utilities,  the resulting
equilibria can be made efficient or fair.\\
- Last, we show that the obtained equilibrium is always pure: after
convergence, each user is associated to a single technology.

To validate our results, we propose several practical implementations
of the algorithm and assess their performance in the practical setting
of a geographical area covered by a global WiMAX network overlapping
with several local IEEE 802.11 (also called WiFi) cells. We suppose
that each user can multi-home, that is to say split her traffic
between her local WiFi network and the global WiMAX cell, in order to
maximize her repercussion utility (to be defined later).

The integration of WiFi and UMTS or WiFi and WiMAX technologies has
already received some attention in the past.

There is a family of papers looking for solutions using Markov or
Semi-Markov Decision Processes \cite{kumar06, coupechoux08}. Based on
Markovian assumptions upon the incoming traffic, these works provide
with numerical solutions, so as to optimize some average or discounted
reward over time. Yet, because of the complexity of the system at hand
(the equations of the throughput in actual wireless systems are not
linear, and not even convex), important simplifying assumptions need
to be made, and the size of the state space quickly becomes
prohibitive to study real systems. Moreover, these methods require to
precisely know the characteristics of the system (e.g. in terms of
bandwidth achieved in all configurations, interference impact of one
cell over the neighboring ones, rate of arrivals), data that are
hardly available in practice.

Our approach is rather orthogonal as we seek algorithms that converge
towards an efficient allocation, using real-time measurements rather
than off-line data.  Such an approach follows game theory
frameworks. There has been recent work that, based on evolutionary
games~\cite{shakkottai07}, provide with optimal
equilibria. Evolutionary games~\cite{weibull97, hofbauer03}, or the
closely-related population games, are based on Darwinian-like
dynamics.  The evolutionary game literature is now mature and includes
several so-called population dynamics, which model the evolution of
the number of individuals of each population as time goes by. In our
context, a population can be seen as a set of individuals adopting the
same strategy (that is to say choosing the same network cell in the
system and adopting identical network parameters). Recent
work~\cite{shakkottai07} have shown that, considering the so-called
\emph{replicator dynamics}, an appropriate choice of the fitness
function (that determines how well a population is adapted to its
environment) leads to efficient equilibria. However they do not provide with algorithms that follow
the replicator dynamics (and hence converge to the
equilibria). Additionally they do not justify the use of evolutionary
games. Indeed, such games assume a large number of individuals, each
of them having a negligible impact on the environment and the fitness
of others. This assumption is not satisfied here, where the number of
active users in a given cell is on the order of a few tens. The
arrival or departure of a single one of them hence significantly
impacts the throughput allocated to others. As the number of players
is limited, we are hence dealing with another kind of equilibria,
namely the Nash Equilibria.

The third trend of this article concerns Nash equilibria learning
mechanisms. In the context of load balancing, a few algorithms (see,
for instance~\cite{barth08, sastry94}) have been shown to converge to
Nash Equilibria. Interestingly enough, it has been pointed out that
this class of algorithms has similar behavior and convergence
properties as replicator dynamics in evolutionary game theory.  It is
to be noted that the main weakness of these algorithms is that they
may converge to \emph{mixed} strategy Nash equilibria, that is to say
to equilibria where each user randomly picks up a decision at each
time epoch. Such equilibria are unfortunately not interesting in our
case, as they amount to perpetual handover between networks.

Finally, there is a growing interest in measuring or analyzing the
efficiency of Nash Equilibria. The most famous concept is certainly
the ``price of anarchy''~\cite{papa}. Let us also mention the more
recent SDF (Selfish Degradation Factor)~\cite{info06}. We will show in
the following that the Nash Equilibria our algorithm converges to are
locally optimal with respect to these two criteria.  In
addition, it has interesting fairness properties. Indeed, we show how
our algorithm can be tuned so as to converge to $\alpha$-fair points
(defined in cooperative game theory, see~\cite{mo}), for arbitrary
value of the parameter $\alpha$. This wide family of fairness criteria
includes in particular the well-known max-min fairness and
proportional fairness and can be generalized so as to cover the Nash
Bargaining Solution point~\cite{equite}.

In the present paper, we hence propose to make use of the previous
works in evolutionary games on heterogeneous network, with additional
fairness considerations, while proposing methods based on works on
Nash learning algorithms that can be implemented on future mobile
equipments. In addition, our work present a novel result which is that
our algorithm converges to \emph{pure} (as opposed to mixed)
equilibria, preventing undesired repeated handovers between stations.

\section{Framework and Model}

In this section, we present the model and the objective of this work
while introducing the notations used throughout the paper.

\subsection{Interconnection of Heterogeneous wireless networks}

We consider a set $\mathcal{N}$ of mobiles, such as mobile $n$ can
connect to a set of network cells, that can be of various technologies
(WiFi, WiMAX, UMTS, LTE...). The set of cells that users\footnote{In the
  following we use the term \emph{users} and \emph{mobiles}
  interchangeably.} can connect to, depends on their geographical
location, wireless equipment and operator subscription.

\subsection{User throughput and cell  load}

By \emph{throughput}, we refer to the rate of useful information
available for a user, in a given network, sometimes also called
\emph{goodput} in the literature.

The throughput obtained by an individual on a network depend on both
her own parameters and the ones of others. These parameters include
geographical position (interference and attenuation level) as well as
wireless card settings (coding schemes, TCP version, to cite a
few). In previous papers~\cite{kumar06, coupechoux08}, the authors
discretize the cells of networks into \emph{zones} of identical
throughput (see Fig~\ref{fig:system}). This means that users in the
same zone will receive the same throughput. Here, we can consider that
each user is in its  own zone\footnote{unlike in the cited papers,
  we are not constrained by the size of the system that is increasing
  with the number of zones.}. The set of users connected to a network
is called the \emph{load}
of the network.

\begin{figure}[htb]
\begin{center}
\includegraphics[scale=.5]{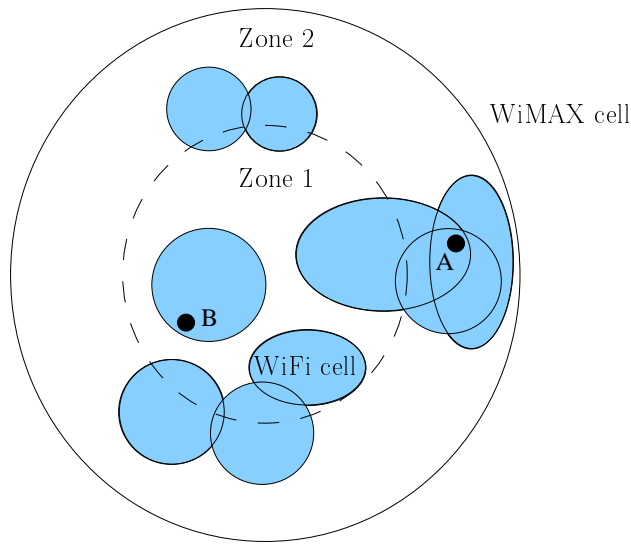}
\end{center}
\caption{An heterogeneous wireless system consisting of a single MAN
  (Metropolitan Area Network, e.g. WiMAX) cell and a set of partly
  overlapping LAN (Local Area Networks, e.g. WiFi) hot-spots (in
  grey). As user $B$ (in zone 1) is closer to the WiMAX antenna, it
  can use a more efficient coding scheme than $A$ (in zone 2) (for
  instance QAM 16 instead of QPSK). Zones are represented with a dash
  line, as opposed to cells, with full lines.}
\label{fig:system}
\end{figure}

More formally, we suppose that each user has a set of network cells
she can connect to denoted by $\mathcal{I}_n$. An action $s_n$ for
user $n$ is the choice of a cell $i \in \mathcal{I}_n$. Then, we
denote by $s$ the vector of users actions $s=(s_n)_{n \in
  \mathcal{N}}$, and call it an \emph{allocation} of mobiles to
networks.

Then, for each allocation $s$, the \emph{load} of network $i$ is
denoted by $\ell^i(s) \in \{0;1 \}^N $, and is such that
$\ell^i_n(s)=1$ if user $n$ takes action $i$, $0$ otherwise. The
throughput $u_n(\ell^i(s))$ of user $n$ taking action $i$ is a
function depending only of the vector of load of cell $i$.  With these
notations, the throughput received by $n$ when she takes decision
$s_n$ is $u_n(\ell^{s_n}(s))$.

\subsection{Pure versus mixed strategies}

As opposed to multi-homing between WiFi systems (see
\cite{shakkottai05, shakkottai07}), multi-homing between different
technologies (e.g. WiFi and WiMAX) induces several complications: the
different technologies may have different delays, have different
packet sizes or coding systems,... and re-constructing the messages
sent by the mobiles may be hazardous. Hence, while each user can
freely switch between the networks cells she has access to, we aim at
algorithms that converge - after a transitional state - to equilibria
in which each user uses a single network (so as to avoid the
cumbersome handover procedure). These are called \emph{pure strategy
  equilibria} (see Section~\ref{sec:potential}).

Yet, during the convergence phase, each mobile is using \emph{mixed
  strategies}\footnote{The formal definition is given in
  Section~\ref{sec:potential}. 
}. Then, the experienced throughput needs to be considered in terms of
expectations. In this case, $q_n$ is a vector of probabilities where
$q_{n,i}$ is the probability for mobile $n$ chooses cell $i \in
\mathcal{I}_n$. The global strategy set is the matrix $q = (q_n)_{n
  \in \mathcal{N}}$, while the choice $S_n(q)$ is now a random variable such that
$\mathbb{P}(S_n(q)=i)=q_{n,i}$. It follows that the expected
throughput received by user $n$ is $\mathbb{E}[u_n(\ell^{S_n(q)}(S(q)))]$,
where $S(q)=(S_n(q))_{n\in \mathcal{N}}$.

\subsection{Efficiency and Fairness}

In our approach, we consider elastic or data traffic. Then, the
Quality-of-Service (QoS) experienced by each mobile user is its
experienced throughput. We are hence interested in seeking equilibria
that are optimal (in the sense of Pareto) in terms of throughput.
Such equilibria is a strategy $q$ such that one cannot find another
strategy $q'$ that increases the expected throughput of a user without
decreasing that of another one: $\forall q' \neq q, \exists n \in
\mathcal{N}$ s.t. $ \mathbb{E}[u_n(\ell^{S_n(q')}(S(q')))] >
\mathbb{E}[u_n(\ell^{S_n(q)}(S(q)))] \Rightarrow \exists m \in
\mathcal{N}, \mathbb{E}[u_m(\ell^{S_m(q')}(S(q')))] <
\mathbb{E}[u_m(\ell^{S_m(q)}(S(q)))].  $

We design a fully distributed algorithm that converges to points which
are not only Pareto optimal but also {\it $\alpha$- fair}. The class
of $\alpha$-fair points~\cite{mo}, 
achieves
\begin{equation}
  \max_{q} \sum_{n \in \mathcal{N}}
  \mathbb{E}[G_\alpha(u_n(\ell^{S_n(q)}(S(q))))]
  \mbox{ with }
  G_\alpha(x) \bydef \frac
  {x^{1-\alpha}}{1-\alpha}.
\label{eq:alpha}
\end{equation}

In the case of pure strategies, for each mobile $n$ such that $S_n=i$,
$\mathbb{E}[u_n(\ell^{i}(S))]=u_n(\ell^{i}(s)) $. So, we aim at
building an algorithm that converges to an allocation $s^*$ that reaches
$$ 
\max_{s} \sum_{n \in \mathcal{N}} G_\alpha(u_n(\ell^{i}(s))).
$$

When $\alpha=0$, the corresponding solution is a social optimum. When
$\alpha$ tends to one, the solution is a proportional fair point (or
Nash Bargaining Solution) and when $\alpha$ tends to infinity, it
converges to a max-min fair point. The parameter $\alpha$ hence allows
flexibility in choosing between fully efficient versus fair
allocation, while ensuring Pareto optimality.

Finally, it is well-known that selfish behavior in the use of
resources (networks) may lead to inefficient use, in case of
congestion for example. To circumvent this, we introduce some
repercussion utility functions that are notified to users. Thus,
instead of competing for throughput, we consider an algorithm
reflecting a non-cooperative game between users that compete for%
maximizing their repercussion utility.  We will give an explicit
closed-form of the repercussion utility function in
Section~\ref{s-reper}. As in the throughput case, the repercussion
utility on a cell only depends on the load on that cell. We denote by
$r_n(\ell^{s_n}(s))$ the repercussion utility received by user $n$ (as
for the throughput, the repercussion utility received by that user
also depends on the choices of the other mobiles of the system, as
reflected in the allocation vector $s$). In the case of mixed
strategies, the expected repercussion utility is
$\mathbb{E}[r_n(\ell^{S_n}(S))]$. The study of such games is given in
the next section.

\section{Allocation Games Related to Potential Games}

This section is devoted to the formal study of allocation games. After
defining what is an allocation game in Section~\ref{s-alloc}, we
introduce the repercussion utilities in Section~\ref{s-reper} what
leads to a new game that is characterized in
Section~\ref{s-char}. Finally, we show the useful property that this
game is a potential game (Section~\ref{sec:potential}).

\subsection{Allocation Games}
\label{s-alloc}

We consider a normal-form game $(\mathcal{N},\mathcal{I},\mathcal{U})$
consisting of a set $\mathcal{N}$ of players ($|\mathcal{N}|=N$),
player $n$ taking actions in a set $\mathcal{I}_n\subset\mathcal{S}$
($|\mathcal{I}_n|=I_n$), where $\mathcal{S}$ is the set of all
actions. Let us denote by $s_n \in \mathcal{I}_n$ the action taken by
player $n$, and $s=(s_n)_{n\in \mathcal{N}} \in {\mathcal I} =
\bigotimes_{n=1}^N\mathcal{I}_n $. Then, $\mathcal{U}=(U_n)_{n \in \mathcal{N}}$ refers to the
\emph{utility} or \emph{payoff} for each player: the payoff for player
$n$ is $U_n(s_1,\dots, s_n, \dots, s_N)$.

By definition, an \emph{allocation game} is a game such that the
payoff of a player when she takes action $i$ only depends on the set
of players who also take action $i$. One can interpret such a game as
a set of users who share a common set of resources ${\mathcal S}$, and
an action vector corresponds to an allocation of resources to users
(hence the name of these games).

We define the {\it load} on action (or resource) $i$ by $\ell^i(s) \in
\{0;1\}^N$ as a vector such that $\ell^i_n(s)=1$ if player $n$ take
action $i$, $0$ otherwise. When there is no ambiguity, we will
simplify the notation and use $\ell=\ell^i(s)$. We denote by
$\ell^{s_n}(s)$ the load on the action taken by player $n$, and we
denote the payoff for player $n$ by $u_n(\ell^{s_n}(s)) \bydef
U_n(s_1,\dots, s_n, \dots, s_N)$.

Hence, allocations games are a wider class of games than
\emph{congestion games} where the payoff of each player depends on the
\emph{number} of players adopting the same strategy
\cite{rosenthal73}. They represent systems where different users
accessing a given resource may have a different impact.

\subsection{Repercussion utilities}
\label{s-reper}
We build a companion game of the allocation game, denoted $({\mathcal
  N},{\mathcal I},{\mathcal R})$. The new player utilities, called
\emph{repercussion utilities} are built from the payoffs of the
original game, according to the following definition:

\begin{definition}[allocation game with repercussion utilities]
  Let us consider the repercussion utility for player $n$ to be:
\begin{multline*}
  r_n(\ell^{s_n}(s)) \bydef u_n(\ell^{s_n}(s)) - \sum_{m\neq n
    :s_m=s_n} \left(u_m(\ell^{s_m}(s)-e_n)- u_m(\ell^{s_m}(s))
  \right),
\end{multline*}
where $e_n$ denotes the vector whose entries are all $0$ but the
n${}^{th}$ one, which equals 1.

An \emph{allocation game with repercussion utilities} is a game whose
payoffs are repercussion utilities.
\end{definition}

The utilities defined in this manner have a natural interpretation: it
corresponds to the player's payoff ($u_n(\ell^{s_n}(s))$) minus the
total increase in payoff for all users impacted by the presence of a
given user on a given commodity ($\displaystyle \sum_{\substack{m\neq
    n :\\s_m=s_n}} \left[u_m(\ell^{s_m}(s)-e_n)- u_m(\ell^{s_m}(s))
\right]$). This is more obvious in the following equivalent formulation.

\begin{remark}
  An equivalent formulation of the repercussion utilities is:
$$
r_n(\ell^{s_n}(s)) =\displaystyle
\sum_{m:\ell^{s_n}_m=1}u_m(\ell^{s_n}(s)) - \sum_{\substack{m\neq n:
    \ell^{s_n}_m(s)=1}}u_m(\ell^{s_n}(s)-e_n).
$$
\end{remark}

\subsection{Characterization of Allocation Games with Repercussion Utilities}
\label{s-char}

We now give a characterization of a payoff that is a repercussion
utility.

\begin{proposition}\label{propositionreward}
An allocation game $({\mathcal N},{\mathcal I},{\mathcal R})$ is an
allocation game with repercussion utilities if
  and only if $ \forall \ell, \forall n,m \in \mathcal{N} \textrm{
    s.t. } s_m=s_n$,
\begin{equation} 
 r_{n}(\ell)-r_{n}(\ell-e_m) =r_{m}(\ell)-r_{m}(\ell-e_n).
\label{eq:symetrie}
\end{equation}
\end{proposition}

\begin{proof}
  Suppose that $r$ is a repercussion utility, then there exists a
  payoff $u$ such that:
$$ r_{n}(\ell) = \sum_{\ell_k=1}u_{k}(\ell) - \sum_{k\neq
  n:\ell_k=1}u_{k}(\ell-e_n).$$ Then, denote
$$ A = \left(\sum_{\ell_k=1}u_{k}(\ell-e_m) - \sum_{k\neq
    n:\ell_k=1}u_{k}(\ell-e_n-e_m)\right).$$ Then,
\begin{multline*}
$$
\begin{array}{ll}
  r_{n}(\ell)-r_{n}(\ell-e_m) &
  = \displaystyle \sum_{\ell_k=1}u_{k}(\ell) - \sum_{k\neq
    n:\ell_k=1}u_{k}(\ell-e_n) - A\\
  & =\displaystyle \sum_{\ell_k=1}u_{k}(\ell) - \sum_{k\neq
    m:\ell_k=1}u_{k}(\ell-e_m) - A\\
  & =r_{m}(\ell)-r_{m}(\ell-e_n). 
\end{array}
$$
\end{multline*}

Conversely, consider an allocation game $(\mathcal{N}, \mathcal{I},
\mathcal{R})$ such that Eq.~\ref{eq:symetrie} is satisfied. Consider
an action $i$ and $\ell$ the vector of load on action $i$. Let $K
\bydef \sum_{n \in \mathcal{N}} \ell_n$ is the number of players
taking action $i$. Further, let $(a(k)), 1 \leq k \leq K$ be the
subscripts of all players taking action $i$. If there are $K$ such
players, then $\ell = \sum_{k=1}^{K} e_{a(k)}$.  Then, we claim that,
for any permutation $\sigma$ of $\{1,\dots,K\}$:
\begin{equation}
  \sum_{k=0}^{K-1} r_{a(k+1)}(\ell-\sum_{j=1}^{k} e_{a(j)}) = 
  \sum_{k=0}^{K-1} r_{a(\sigma(k+1))}(\ell-\sum_{j=1}^{k}
  e_{a(\sigma(j))}).
\label{eq:permutations}
\end{equation}
Indeed, note that, from Eq.~\ref{eq:symetrie}:
\begin{multline*}
  r_{a(k+1)}\left(\ell-\sum_{j=1}^{k-1} e_{a(j)}\right)
  - r_{a(k+1)}\left(\ell-\sum_{j=1}^k e_{a(j)}\right) =\\
  r_{a(k)}\left(\ell-\sum_{j=1}^{k-1} e_{a(j)}\right) -
  r_{a(k)}\left(\ell-\sum_{j=1}^{k-1} e_{a(j)}-e_{a(k+1)}\right).
\end{multline*}
Therefore:
\begin{multline*}
  r_{a(k)}\left(\ell-\sum_{j=1}^{k-1} e_{a(j)}\right) +
  r_{a(k+1)}\left(\ell-\sum_{j=1}^k
  e_{a(j)}\right) = \\
  r_{a(k+1)}\left(\ell-\sum_{j=1}^{k-1} e_{a(j)}\right) +
  r_{a(k)}\left(\ell-\sum_{j=1}^{k-1} e_{a(j)}-e_{a(k+1)}\right).
\end{multline*}
Hence, for any $k$, the sum $\sum r_{a(k+1)}(\ell-\sum_{j=1}^k
e_{a(j)})$ remains unchanged if one swaps $a(k)$ and $a(k+1)$
(elementary transposition). Then, Eq.~\ref{eq:permutations}
results from the fact that any permutation $\sigma$ can be decomposed
in a finite number of elementary transpositions.

We now construct a payoff $u$ as follow: for any $n$ such that
$\ell_n=1$, let us define:
$$\displaystyle
u_n(\ell) \bydef \frac{1}{K} \sum_{k=0}^{K-1} r_{a(k+1)}(\ell-\sum_{j=1}^k
e_{a(j)}).
$$
Then, 
\begin{multline*}
  \displaystyle \sum_{\ell_m=1}u_{m}(\ell) - \sum_{m\neq
    n:\ell_m=1}(u_{m}(\ell-e_n)) =\\
  \sum_{k=0}^{K-1} r_{a(k+1)}(\ell-\sum_{j=1}^k e_{a(j)}) -
  \sum_{k=0}^{K-2} r_{b(k+1)}(\ell-e_n-\sum_{j=1}^k e_{b(j)}).
\end{multline*}

Note that the sequence $a$ is identical to sequence $b$ with the
additional element $n$. From~Eq.~\ref{eq:permutations}, we can choose
a permutation $\sigma$ such that $a(\sigma(1))=n$. Then:

$$
\begin{array}{l}
  \displaystyle \sum_{\ell_m=1}u_{m}(\ell) - \hspace{-0.8em} \sum_{m\neq
    n:\ell_m=1}\hspace{-0.8em}(u_{m}(\ell-e_n))\\
  = \displaystyle \sum_{k=0}^{K-1} r_{a(\sigma(k+1))}(\ell-\sum_{j=1}^k
  e_{a(\sigma(j))}) -\sum_{k=1}^{K-1} r_{a(\sigma((k+1))} (\ell-e_n-\sum_{j=2}^k
  e_{a(\sigma(j))}) \\
  = \displaystyle \sum_{k=1}^{K-1}
  r_{a(\sigma(k+1))}(\ell-e_{a(\sigma(1))}-\sum_{j=2}^k
  e_{a(\sigma(j))}) + \displaystyle r_{a(\sigma(1))}(\ell) -\\ 
  \hfill \displaystyle \sum_{k=1}^{K-1} r_{a(\sigma((k+1))}(\ell-e_n-\sum_{j=2}^k
  e_{a(\sigma(j))})\\
  = r_n(\ell).
\end{array}
$$

Hence $(\mathcal{N},\mathcal{I},\mathcal{R})$ is the allocation game
with repercussion utilities associated to the
$(\mathcal{N},\mathcal{I},\mathcal{U})$ allocation game.
\end{proof}

From Prop.~\ref{propositionreward}, we conclude that allocation games
with repercussion utilities are a special subset of allocation
games. The results presented in the following are hence valid for any
allocation game such that Eq.~\ref{eq:symetrie} is satisfied.

\begin{example}
  Let $M$ be the payoff matrix of a two-player game. This amounts to
  saying that the first (resp. second) player chooses the line and the
  second chooses the column. The payoff for the first player is given
  by the first (resp. second) component.
$$M = 
\left(
\begin{array}{cc}
(a,A)&(b,B)\\
(c,C)&(d,D)
\end{array}
\right).
$$
It follows from Proposition~\ref{propositionreward} that this is a
game with repercussion utilities if and only if $a=A+b-C$ and
$d=D+c-B$. Then, one can check the interesting property that there
necessarily exists a pure Nash equilibrium (for instance $(a,A)$ is a
Nash equilibrium if $a \geq c$ and $A \geq B$).
\end{example}

\subsection{Allocation Games with Repercussion Utilities are Potential Games}
\label{sec:potential}

In this section, we show that, given an allocation game, the game with
repercussion utilities (1) admits a potential function and (2) this
potential equals the sum of the payoffs for all players in the initial
game.
This appealing property is exploited in the next section to show some
strong results on the behavior of the well-known replicator dynamics
on such games.

Consider an allocation $({\mathcal N},{\mathcal I},{\mathcal U})$ and
its companion game $({\mathcal N},{\mathcal I},{\mathcal R})$. We
first assume that players have \emph{mixed} strategies. Hence a
strategy for player $n$ is a vector of probability
$q_n=(q_{n,i})_{i\in \mathcal{I}_n }$, where $q_{n,i}$ is the
probability for player $n$ to take action $i$ (i.e.  $q_{n,i}\geq 0$
and $ \sum_{i\in \mathcal{I}_n}q_{n,i}=1$). The strategy domain for
player $n$ is $\Delta_n \bydef \{ 0 \leq q_{n,i} \leq 1, s.t.
\sum_{i\in \mathcal{I}_n}q_{n,i}=1\}$. Then, the global
domain\footnote{Notice that $\Delta$ is a polyhedron.} is $\Delta =
\bigotimes_{n=1}^N \Delta_n$ and a global strategy is $q \bydef (q_n)_{n
  \in \mathcal{N}}$. We say that $q$ is a {\it pure strategy} if for
any $n$ and $i$, $q_{n,i}$ equals either 0 or 1.

We denote by $S$ the random vector whose entries $S_n$
are all independent and whose distribution is $\forall n \in
\mathcal{N}, \forall i \in \mathcal{I}_n, \,\mathbb{P}(S_n=i)=q_{n,i}$.
The expected payoff for player $n$ when she takes action $i$ is
$f_{n,i}(q) \bydef \mathbb{E}[r_n(\ell^i(S))|S_n=i]$. Then, her mean
payoff is $ \overline{f}_n(q) \bydef \displaystyle \sum_{i \in
  \mathcal{I}_n}q_{n,i}f_{n,i}(q)$.  We can notice that $f_{n,i}(q)$
only depends on $(q_{m,i})_{m\neq n} $ and it is a multi-linear function
of $(q_{m,i})_{m\neq n}$.

The next theorem claims that the allocation game with repercussion
utilities is a potential game. Potential games were first introduced
in~\cite{monderer96}. The notion was afterward
extended to continuous set of players~\cite{sandholm01}. In our case, it
refers to the fact that the expected payoff for each player derives
from a potential function. More precisely, we show that
$\displaystyle f_{n,i}(q)=\frac{\partial F}{\partial q_{n,i}}(q)$,
where 
\begin{equation}
F(q) \bydef \sum_{n \in \mathcal{N}}\sum_{i \in
  \mathcal{I}_n}q_{n,i}\mathbb{E}[u_n(\ell^i(S))|S_n=i].
\label{eq:potentiel}
\end{equation}
It is interesting to notice the connection between $f_{n,i}(q)$ which
is the expected repercussion utility, and $F(q)$ which refers to the
sum of expected payoffs in the initial game. A strategy that increases
the expected repercussion utility of a player, yields to a marginal
increase of the potential.

\begin{theorem}
  The allocation game with repercussion utilities is a potential game,
  and its associated potential function is $F$, as defined in
  Eq.~\ref{eq:potentiel}.
\end{theorem}

\begin{proof}
  Let us first differentiate function $F$:
$$
\frac{\partial F}{\partial q_{n,i}}(q)=
\mathbb{E}[u_n(\ell^i(S))|S_n=i]+ \\ \sum_{m\neq
  n}q_{m,i}\frac{\partial \mathbb{E}[u_m(\ell^i(S))|S_m=i] }{\partial
  q_{n,i}}.
$$
In fact, it is clear that $\displaystyle \frac{\partial
  \mathbb{E}[u_m(\ell^j(S))|S_m=j] }{\partial q_{n,i}}=0$ if $j \neq
i$, and $\displaystyle \frac{\partial \mathbb{E}[u_n(\ell^i(S))|S_n=i]
}{\partial q_{n,i}}=0$.
To simplify the notations, we omit the index $i$. Then,
$$
\begin{array}{l}
  \displaystyle \frac{\partial F}{\partial q_{n}}(q)
  \displaystyle=\mathbb{E}[u_n(\ell(S))|S_n=i]+ \sum_{m\neq
    n}q_{m}\frac{\partial }{\partial q_{n}}
  \sum_{\ell}u_m(\ell)\mathbb{P}(\ell(S) = \ell|S_m=i)\\
  \displaystyle=\mathbb{E}[u_n(\ell(S))|S_n=i]+\\
  \hfill \displaystyle\sum_{m\neq
    n}q_{m}\sum_{\ell}u_m(\ell)\bigg(
  \mathbb{P}(\ell(S)=\ell|S_m=i,S_n=i)
  -\mathbb{P}(\ell(S)=\ell|S_m=i,S_n\neq i)\bigg)\\
  \displaystyle=\mathbb{E}[u_n(\ell(S))|S_n=i]+ \\
  \hfill \displaystyle\sum_{m\neq
    n}\sum_{\ell}u_m(\ell) × \bigg(
  \mathbb{P}(\ell(S)=\ell,S_m=i|S_n=i)
  -\mathbb{P}(\ell(S)=\ell+e_n,S_m=i|S_n= i)\bigg)\\
  \displaystyle=\mathbb{E}[u_n(\ell(S))|S_n=i] - \sum_{\substack{m\neq
      n: S_m=S_n}}\bigg(\mathbb{E}[u_m(\ell(S)-e_n)|S_n=i]-
  \mathbb{E}[u_m(\ell(S))|S_n=i]\bigg)\\
  =\displaystyle \mathbb{E}[r_n(\ell(S))|S_n=i]\\
  =\displaystyle f_{n,i}(q).
\end{array}
$$
\end{proof}

\begin{remark}
  By adding a  large constant to all payoff $u$, the repercussion
  utilities become positive. Clearly, this has no impact on the relative
  potential of allocations. The Nash Equilibria of the allocation game
  are also conserved. In the following, we will assume that the
  repercussion utilities are positive.
\end{remark}

\section{Replicator dynamics and algorithms}
\label{sectionapplication}

In this section, we show how to design a strategy update mechanism for
all players in an allocation game with repercussion utilities that
converges to pure Nash Equilibria. We will study in the next section
(Section \ref{s-attract}) their efficiency properties.

\subsection{Replicator Dynamics.}
\label{s-repdy}

We now consider that the player strategies vary over time, hence
$q$ depends on the time $t$: $q=q(t)$. The trajectories of the
strategies are described below by a dynamics called \emph{replicator
  dynamics}. We will see in section \ref{s-approx} that this
dynamics can be seen as the limit of a learning mechanism.

\begin{definition}
  The replicator dynamics ~\cite{weibull97}\cite{hofbauer03} is
  ($\forall n\in \mathcal{N}, i\in \mathcal{I}_n$):
\begin{equation}\label{dynamics}
\frac{dq_{n,i}}{dt}(q)=q_{n,i}\left( f_{n,i}(q)-\overline{f_n}(q)\right).
\end{equation}
We say that $\hat{q}$ is a \emph{stationary point} (or equilibrium point) 
if ($\forall n\in \mathcal{N}, i\in
\mathcal{I}_n$):
$$\frac{dq_{n,i}}{dt}(\hat{q})=0 .$$
\end{definition}
In particular, $\hat{q}$ is a stationary point implies $\forall n\in
\mathcal{N}, i\in \mathcal{I}_n ,\, \,\hat{q}_{n,i}=0$ or
$f_{n,i}(\hat{q})=\overline{f}_n(\hat{q})$.

Intuitively, this dynamics can be understood as an update mechanism
where the probability for each player to choose actions whose expected
payoffs are above average will increase in time, while non profitable
actions will gradually be abandoned.

Let us notice that the trajectories of the replicator dynamics remain
inside the domain $\Delta$. Also, from \cite{sandholm01}, the
potential function $F$ is a strict Lyapunov function for the
replicator dynamics, that means that the potential is strictly
increasing along the trajectories outside the stationary points.

In this context, a closed set $A$ is {\it Lyapunov stable} if, for
every neighborhood $B$, there exists a neighborhood $B'\subset B$ such
that the trajectories remain in $B$ for any initial condition in $B'$.
$A$ is {\it asymptotically stable} if it is Lyapunov stable and is an
attractor (i.e. there exists a neighborhood $C$ such that all
trajectories starting in $C$ converge to $A$). The existence of a
strict Lyapunov function yields the following:


\begin{remark}
The accumulation  points of the trajectories of the replicator dynamics are
stationary points.
\end{remark}

Intuitively, the limit points (that are connected) of the same
trajectory must have the same value for the Lyapunov function. But the
set of limit points is invariant for the dynamics, hence the Lyapunov function
is non-increasing on this set. The remark follows.

\begin{proposition}\label{pure}
  All the asymptotically stable sets of the replicator dynamics are
  faces of the domain. These  faces are sets of equilibrium points for
  the replicator dynamics.
\end{proposition}

\begin{proof}
  We show that any set which is not a face of the domain is not an
  attractor. This results from a property discovered by E. Akin
  \cite{akin83} which states that the replicator dynamics preserves a
  certain form of volume.

  Let $A$ be an asymptotically stable set of the replicator dynamics.
  Since the domain $\Delta$ is polyhedral, $A$ is included in a face
  $F_A$ of $\Delta$. The support of the face $S(F_A)$ is the set of
  subscripts $(n,i)$ such that there exists $q \in A$ with $q_{n,i} \not=
  0 $ or $1$.  The relative interior of the face is ${\mbox Int}(F_A)
  = \{ q \in F(A) s.t. \forall (n,i) \in S(F_A), 0 < q_{n,i} < 1 \} $.

  Furthermore, it should be clear that faces are invariant under the
  replicator dynamics.  Hence on the face $F_A$, by using the
  transformation $\displaystyle v_{n,i}\bydef
  log(\frac{q_{n,i}}{q_{n,i_n}}), \forall q \in {\mbox Int}(F_A)$, one
  can see that
\[ \displaystyle \frac{\partial}{\partial
    v_{n,i}}\frac{dv_{n,i}}{dt}=0, \forall n \in \mathcal{N}, i \in
  \mathcal{I}.\]

  Up to this transformation, the divergence of the vector field is
  null on $F_A$. Using Liouville's theorem~\cite{akin83}, we infer
  that the transformed dynamics preserves volume in ${\mbox
    Int}(F_A)$. This implies that the set of limit points of the
  trajectories in ${\mbox Int}(F_A)$ is ${\mbox Int}(F_A)$ itself. By
  the previous remark, ${\mbox Int}(F_A)$ is made of equilibrium
  points. By continuity of the vector field, all the points in face
  $F_A$ are equilibria. Finally, since $A$ is asymptotically stable,
  this means that $A = F_A$.
\end{proof}

We say that $s=(s_n)_{n\in \mathcal{N}}$ is a \emph{pure Nash
  Equilibrium} if $\forall n \in \mathcal{N}$, $\forall s'_n \neq s_n,
U_n(s_1 \dots s_n \dots s_{N}) \geq U_n(s_1 \dots s'_n
\dots s_{N})$.
\begin{remark}
  Let $q$ be a pure strategy. We denote by $i_n$ the choice of player
  $n$ such that $q_{n,i_n}=1$. Then, a pure strategy $q$ is a Nash
  equilibrium is equivalent to:
$$\forall n \in \mathcal{N},  \forall j \neq i_n, f_{i_n,n}(q) \geq f_{j,n}(q). $$
\end{remark}

The following proposition comes form a classical result that says that
the pure Nash equilibria are asymptotically stable points of the
replicator dynamics.

\begin{proposition}\label{nasheq}
  If a stable face is reduced to a single point, then this of the replicator dynamics are pure
  Nash equilibria of the allocation game with repercussion utilities.
\end{proposition}
\begin{proof}
  Let $\hat{q}$ be an asymptotically stable point.  Then $\hat{q}$ is
  a face of $\Delta$ by Proposition \ref{pure} ({\it i.e.} a 0-1
  point), with, say $\hat{q}_{n,i} =1$.  Assume that $\hat{q}$ is not
  a Nash equilibrium. Then, there exists $j\not= i$ such that
  $f_{j,n}(\hat{q}) \geq f_{i,n}(\hat{q}) $.  Now, consider a point
  $q' = \hat{q} + \epsilon e_{n,j} - \epsilon e_{n,i}$.  Notice that
  $f_{n,i}(q') = f_{n,i}(\hat{q})$ since $q'$ and $\hat{q}$ only
  differ on components concerning user $n$.  Then starting in $q'$,
  the replicator dynamics is
  \begin{equation*}
    \begin{array}{l@{\,}l@{\,}l}
      \displaystyle \frac{dq_{n,i}}{dt}(q') &=& \displaystyle q'_{n,i}( f_{n,i}(q')
      - ((1-\epsilon)
      f_{n,j}(q')  + \epsilon  f_{n,i}(q') ) \\[0.6em]
      &=& \displaystyle (1-\epsilon)( f_{n,i}(\hat{q}) - ((1-\epsilon)
      f_{n,j}(\hat{q}) +
      \epsilon f_{n,i}(\hat{q}) ) \\[0.6em]
      &=& \displaystyle -\epsilon (1-\epsilon) ( f_{n,j}(\hat{q}) -
      f_{n,i}(\hat{q}) ) \\[0.6em] &\leq& 0,
    \end{array}
  \end{equation*}

  and $\displaystyle \frac{dq_{n,j}}{dt}(q') = -
  \frac{dq_{n,i}}{dt}(q') \geq 0.  $

  For all users $m \not=n$, $\forall u \in {\mathcal I}_m, q'_{m,u}
  \in \{0,1\}$, then \\$ \displaystyle \frac{dq_{m,k}}{dt}(q') =
  q'_{m,k}\left( f_{n,k}(q') - \sum_{u} q'_{m,u}( f_{n,u}(q')) \right)
  = 0.$

  Therefore starting from $q'$, the dynamics keeps moving in the
  direction $e_{n,j}-e_{n,i}$ (or stays still) and does not converge
  to $\hat{q}$.  This contradicts the fact that $\hat{q}$ is
  asymptotically stable.
\end{proof}

\begin{proposition}
  Allocation games with repercussion utilities admit at least one pure
  Nash equilibrium.
\end{proposition}
\begin{proof}
  Allocation games with repercussion utilities admit a potential that
  is a Lyapunov function of their replicator dynamics.  Since the
  domain $\Delta$ is compact, the Lyapunov function reaches its
  maximal value inside $\Delta$.  The argmax of the Lyapunov function
  form an asymptotically stable sets $A$ of equilibrium points.  By
  Proposition~\ref{pure}, these  sets  are  faces of the domain (hence
  contain  pure points).  All points in $A$ are Nash equilibrium
  points by using a similar argument as in Proposition \ref{nasheq}.
  This concludes the proof.
\end{proof}

\subsection{A Stochastic Approximation of the Replicator Dynamics.}\label{s-approx}

In this section, we present an algorithmic construction of the
players' strategies that selects  a pure Nash equilibrium
for the game with repercussion utilities.
A similar learning mechanism is proposed in~\cite{sastry94}.
We now assume a discrete time, in which at each epoch $t$, players take
random decision $S_n(t)$ according to their strategy $q_n(t)$, and
update their strategy profile according to their current payoff. We look at
the following algorithm ($\forall n\in \mathcal{N}, i\in
\mathcal{I}_n$):
\begin{equation}\label{algorithm}
q_{n,i}(t+1)=q_{n,i}(t)+\epsilon \, r_n(\ell^{S_n}(S))\,(1_{S_n=i}-q_{n,i}(t)),
\end{equation}
where $S_n=S_n(t)$, $\epsilon > 0$ is the constant step size of the
algorithm, and $1_{S_n=i}$ is equal to $1$ if $S_n=i$, and $0$
otherwise. Recall that we assume that $r_n(\ell^{S_n}(S))\geq
0$. Then, if $\epsilon$ is small enough, $q_{n,i}$ remains in the
interval $[0;1]$. Strategies are initialized with value
$q(0)=q_0$. The step-size is chosen to be constant in order to have
higher convergence speed than with decreasing step size.

One can notice that this algorithm is fully distributed, since for
each player $n$, the only information needed is
$r_n(\ell^{S_n}(S))$. Furthermore, at every iteration, each player
only need the utility on one action (which is randomly chosen). In
applicative context, this means that a player does not have to scan
all the action before update her strategy, what would be costly.

Below, we provide some intuition on why the algorithm is characterized
by a differential equation, and how it asymptotically follows the
replicator dynamics~(\ref{dynamics}).  Note that we can re-write
(\ref{algorithm}) as:
$$q_{n,i}(t+1)=q_{n,i}(t)+\epsilon \, b(q_{n,i}(t),S_n(t)) .$$
Then, we can split $b$ into its expected and martingale components:
$$
\begin{array}{l @{\,=\,}  l}
  \displaystyle \overline{b}(q_{n,i}(t))&\displaystyle\mathbb{E}[b(q_{n,i}(t),S_n(t))]\\
  \displaystyle\nu(t)&\displaystyle b(q_{n,i}(t),S_n(t))-\overline{b}(q_{n,i}(t)) .
\end{array}
$$
Again, (\ref{algorithm}) can be re-written as:
$$\frac{q_{n,i}(t+1)-q_{n,i}(t)}{\epsilon}=\overline{b}(q_{n,i}(t))+\nu(t). $$
As $\nu(t)$ is a random difference between the update and its
expectation, then by application of a law of large numbers, for small
$\epsilon$, this difference goes to zero. Hence, the trajectory of
$q_{n,i}(t)$ in discrete time converges to the trajectory in
continuous time of the differential equation:
\begin{equation*} \left\{
\begin{array}{ll}
\displaystyle \frac{dq_{n,i}}{dt}=\overline{b}(q_{n,i}), & \text{and}\\[0.8em]
q(0)=q_0.
\end{array} \right.
\end{equation*}

Let us compute $\overline{b}(q_{n,i})$ (for ease of notations, we
omit the dependence on time $t$):
$$
\begin{array}{ll}
  \displaystyle \overline{b}(q_{n,i})&\displaystyle =\mathbb{E}[b(q_{n,i},S_n)]\\
  &\displaystyle  =q_{n,i}(1-q_{n,i})f_{n,i}(q)-\sum_{j\neq
    i}q_{n,j}q_{n,i}f_{n,j}(q) \\
  &\displaystyle  =q_{n,i}(f_{n,i}(q)-\sum_jq_{n,j}f_{n,j}(q) ) \\
  &\displaystyle  =q_{n,i}(f_{n,i}(q)-\overline{f}(q) ).  
\end{array}
$$
Then, $q_{n,i}(t)$ follows the replicator dynamics.

Consider a typical run of algorithm (\ref{algorithm}) over a system made of
$10$ users with $5$ choices over $10$ networks. The figure displays
for one user, the probabilities of choosing each of the $5$ possible
choices. As user has $5$ possible choices, at time epoch $0$, each
choice has probability $0.2$. Then, as $t$ grows, all the
probabilities except one, tend to $0$.

\begin{figure}[htb]
\begin{center}
\includegraphics[width=.8\linewidth]{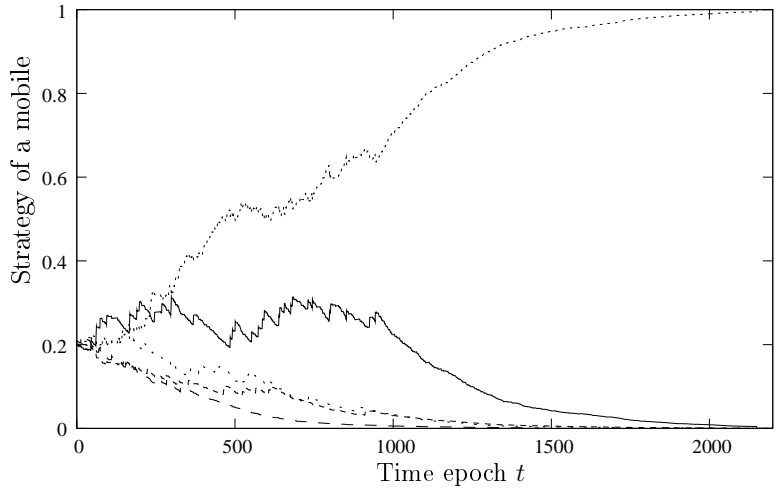}
\end{center}
\caption{Convergence of the probability values for each of the $5$
  possible choices of one user.}
\label{fig:pures}
\end{figure}

\subsection{Properties of the algorithm.}

The algorithm is designed so as to follow the well-known replicator
dynamics. Furthermore, the stochastic aspect of the algorithm provides
some stability to the solution: whereas the deterministic solution of
a replicator dynamics may converge to a saddle point, this cannot
happen with the stochastic algorithm. The use of repercussion
utilities provides a potential to the companion game and it is known
that the potential is a Lyapunov function for the replicator dynamics,
hence the potential is increasing along the trajectories. The
following theorem aggregates the main results about the algorithm
applied on repercussion utilities.

\begin{theorem}\label{thm:recap}
  The algorithm~(\ref{algorithm}) weakly converges to a set of pure
  points that are locally optimal  for the potential function, and
  Nash equilibria of the allocation game with repercussion utilities.
\end{theorem}

\begin{proof}
\begin{itemize}

\item The algorithm is a stochastic algorithm with constant step
  size. From Theorem 8.5.1 of Kushner and Yin \cite{kushner97}, we
  infer that the algorithm weakly converges as $\epsilon \to 0$ to the
  limit points of the trajectories of an ode, which is, in our case,
  the replicator dynamics (\ref{dynamics}) (it is a particular case of
  the theorem in which conditions of the theorem hold: all variables
  are in a compact set and the dynamics is continuous). Furthermore,
  the set to which the sequence $q(t)$ converges is an asymptotically
  stable set of the replicator dynamics, because unstable equilibria
  are avoided (the noise verify condition of \cite{duflo96}, Theorem
  1). From Proposition \ref{pure}, the only asymptotically stable sets
  of the dynamics are faces. Hence the algorithm converges to faces
  which are asymptotically stable.

\item We now show that the dynamics in such a face (denoted by $F$)
  converges almost surely to a pure point.  Let $\hat{q}(0) \in
  F$. Then, the trajectory $\hat{q}(t)$ following the algorithm stays
  in $F$. Furthermore:
  $$
  \begin{array}{l}
    \displaystyle \mathbb{E}[\hat{q}_{n,i}(t+1)|\hat{q}(t)]\\ 
    \displaystyle =\hat{q}_{n,i}(t)(\hat{q}_{n,i}(t)+\epsilon
    f_{n,i}(\hat{q}(t))(1-\hat{q}_{n,i}(t)))\\
    \hspace{5cm}\displaystyle +\sum_{j \neq i}\hat{q}_{n,j}(\hat{q}_{n,i}(t)- \epsilon
    f_{n,j}(\hat{q}(t))\hat{q}_{n,i}(t))\\
    \displaystyle =\hat{q}_{n,i}(t)+\epsilon
    q_{n,i}(f_{n,i}(\hat{q}(t))-\overline{f}_n(\hat{q}(t))).
  \end{array}
$$
Since at a mixed stationary point
$f_{n,i}(\hat{q})=\overline{f}_n(\hat{q})$, then
$\mathbb{E}[\hat{q}_{n}(t+1)|\hat{q}(t)]=\hat{q}_n(t)$. Hence the
process $(\hat{q}_n(t))_t$ is a martingale, and is almost surely
convergent. The process converges necessarily to a fixed point of the
iteration $ \hat{q}_{n,i}(t+1)=\hat{q}_{n,i}(t)+\epsilon \,
r_n(\ell^{s_n}(s))\,(1_{s_n=i}-\hat{q}_{n,i}(t))$, and the sole fixed
points are pure points (since the step size $\epsilon$ is constant).

\item Let $(q(t))_{t \in \mathbb{N}}$ be the random process given by
  the algorithm. Suppose that it admits a closed set $A$ of limit
  points that contains no pure points, such that $A \subset F$, where
  $F$ is the smallest face of the domain $\Delta$ that contains
  $A$. Assume, for ease of notations that $F=\Delta \cap \{q:
  q_{n,i}=0\}$. By Proposition~\ref{pure}, $F$ is a face of $\Delta$
  that is set of stationary points.

  Denote by $A^\delta$ the $\delta-neighborhood$ of $A$. We suppose
  that $\delta$ is small enough to ensure that $A^\delta$ does not
  contain any pure point (this is possible since $A$ is a closed
  set). Let $\mathcal{A}$ the set of $\omega$ such that $\forall
  \omega \in \mathcal{A}, \forall \delta >0, \forall T \in \mathbb{N},
  \exists t>T \text{ s.t. } q(t)\in A^\delta$. We now show that the
  Lebesgue measure of $\mathcal{A}$, denoted $\mu(\mathcal{A})$, is
  null.  Intuitively, as
  the algorithm goes near the face, the probability that it follows a
  martingale in the face is closed to $1$, and then the trajectory
  will not approach the face.

  Let $\hat{\mathcal{A}}$ be the set of $\omega$ such that the
  martingale (in $F$) $\hat{q}(t)(\omega)$ converges to a pure point
  for every initial condition in $F$. The measure of
  $\hat{\mathcal{A}}$ is $1$.  Let $s(\omega)=\inf\{T:
  \forall\hat{q}(0)\in F, \forall t>T, \hat{q}(t)(\omega)\notin
  A^\delta\}$. $s(\omega)$ is the maximal time such that for every
  initial condition in $F$, the martingale is outside
  $A^\delta$. Since $F$ is compact, it follows that for all $\omega
  \in \hat{\mathcal{A}}$, $s(\omega)$ is finite. Let
  $\hat{\mathcal{A}}(T^+) \subset \hat{\mathcal{A}} $
  (resp. $\hat{\mathcal{A}}(T^-)$) be the set of $\omega$ such that
  $s(\omega)>T$ (resp. $s(\omega)\leq T$). Then,
  $\mu(\hat{\mathcal{A}}(T^+))\to 0$ when $T \to
  \infty$.

\item Let $\delta=\frac{\nu}{k}$, where $\nu>0$ and $k\in
  \mathbb{N^*}$. If a trajectory $q(t)(\omega)$ is such that there
  exists $T$ with $\displaystyle q_{n,i}(T)(\omega)<\frac{\nu}{k}$,
  then there exists a duration $T_k$ such that $\forall t \in
  [T-T_k;T], q_{n,i}(t)(\omega)<\nu$, where $r \bydef \max_n \max_s
  r_n(\ell^{s_n}(s))$. We now show that $\displaystyle T_k \bydef
  \min(T, -\frac{ln(k)}{ln(1-\epsilon r)})$. Indeed, $q_{n,i}(t+1)
  \geq q_{n,i}(t)-\epsilon r q_{n,i}(t)$. Then $q_{n,i}(T) \geq
  q_{n,i}(T-T_k)(1-\epsilon r )^{T_k}$. It follows $\displaystyle
  q_{n,i}(T-T_k) \leq q_{n,i}(T)(1-\epsilon r )^{-T_k}\leq
  \frac{\nu}{k}(1-\epsilon r )^{-T_k}=\nu$.

\item Let $p(\nu,T_k)$ be the probability that $q(t)(\omega)$, at
  distance less than $\displaystyle \delta=\frac{\nu}{k}$ of $F$ at
  time $t_0$, does not follow the martingale $\hat{q}(t)(\omega)$
  defined by $\hat{q}(t_0)(\omega)=proj_F(q(t_0)(\omega))$, during
  time $T_k$ (hence $q(t_0+T_k)(\omega)$ can be inside
  $A^\delta$). Then:
  $$\forall k \in \mathbb{N}, \mu(\mathcal{A})
  \leq
  \mu(\hat{\mathcal{A}}(T_k^+))+p(\nu,T_k)\mu(\hat{\mathcal{A}}(T_k^-)).$$
  Indeed, let $k \in \mathbb{N}$ and $\omega \in \mathcal{A}$. There
  is $T$ with $\displaystyle q_{n,i}(T)(\omega)\leq\frac{\nu}{k}$. For
  simplicity, suppose that $T=T_k$. Then, either $\omega \in
  \hat{\mathcal{A}}(T_k^+)$, either $\omega \in
  \hat{\mathcal{A}}(T_k^-)$, either the complementary set in
  $\mathcal{A}$ whose measure is $0$. If $\omega \in
  \hat{\mathcal{A}}(T_k^-)$, then $q(T)(\omega) \in A^\delta$ with
  probability $p(\nu,T_k)$.

  We now show that, by taking an appropriate $\nu=\nu(k)$, then
  $p(\nu(k),T_k)\to 0$, when $k\to \infty$. This, and
  the fact that $\mu(\hat{\mathcal{A}}(T_k^+))\to 0$ when
  $k\to \infty$ implies that $\mu(\mathcal{A})=0$.

\item Suppose $\omega \in \hat{\mathcal{A}}(T^-)$: let us define
  $d_t=d(q(t),\hat{q}(t))$ the distance between the interior
  trajectory, and the martingale trajectory at time $t$. Then, one can
  check that, under $\omega$, $d_{t+1}\leq d_t(1+\epsilon r)$, with
  probability at least $1-p(d_t)$ with $p(d_t)\bydef d_t\sum_{n=1}^N
  I_n$. Let $K \bydef \sum_{n=1}^N I_n$. Indeed the vector of actions
  $s(q)$ is the same as $s(\hat{q})$ as long as $\omega$ picks the
  same choice for all players. The contrary happens with probability
  $\displaystyle \sum_{n=1}^N\sum_{i=1}^{I_n}|\sum_{k=1}^i
  q_{n,k}(t)-\hat{q}_{n,k}(t)|$. See Figure~\ref{coupling} for an
  illustration of this. Then, the lower bound follows from the
  inequality $|\sum_{k=1}^i q_{n,k}(t)-\hat{q}_{n,k}(t)|< d_t$.

\begin{figure}[htb]
  \begin{center}
    \includegraphics[width=0.6\linewidth]{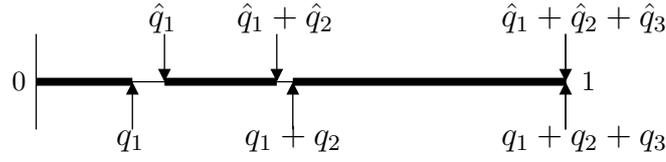}
    \caption{The thick line shows the measure of the set of all $\omega$
      corresponding to the same choices for player  1 (with 3
      choices). \label{coupling}}
  \end{center}
\end{figure}

Since $d_0<\nu$, then $d_{T_k} > \nu (1+\epsilon r)^{T_k} $ with
probability less than $\displaystyle p(\nu,T_k)$ where $\displaystyle
p(\nu,T_k) = 1-\prod_{t=0}^{T_k}(1-p(d_t)) \leq 1 -
\prod_{t=0}^{T_k}(1-K\nu(1+\epsilon r)^t)$. Take $\displaystyle
\nu=\frac{(1+\epsilon r)^{-2T_k}}{K}$. When $k\to \infty$,
$d_{T_k}$ goes to $0$, and $p(\nu,T_k)$ goes to $0$. Hence,
$q(t)(\omega)$ does not follow $\hat{q}(t)(\omega)$ for $t=0$ to
$t=T_k$ with probability $\displaystyle p(\nu,T_k)$, and then can be
inside $A^\delta$.

\item Finally, the fact that the pure point attained is a Nash
  equilibrium follows from Proposition \ref{nasheq}.
\end{itemize}
\end{proof}

One can notice that the convergence of the algorithm to a pure point
relies on the fact that the step size $\epsilon$ is constant. If it
were decreasing, the algorithm would converge to an equilibrium point
in a stable face, that need not be pure.

The combination of both algorithm (\ref{algorithm}) and repercussion
utilities provides an iterative method to select a pure allocation
which is stable, and locally optimal. This can be viewed as a
\emph{selection algorithm}.

\subsection{Global Maximum vs Local Maximum for the Selection Algorithm.}
\label{s-attract}

In the previous section, we showed that the algorithm converges to a
local maximum of the potential function. This induces that if there is
only one local maximum, the algorithm attains the global maximum. This
arises for instance if the potential function is concave. Without the
uniqueness of the local maximum, there is no guaranty of convergence
to the global maximum. Hence, assume there are multiple local maxima
(that are pure points), which is common when the payoffs are
random. Each of them is an attractor for the replicator dynamics. In
this section, we investigate the following question: does the initial
point of the algorithm belongs to the basin of attraction of the
global maximum?

Since every player has no preference at the beginning of the
algorithm, we assume that initially, $\forall n \in \mathcal{N},i \in
\mathcal{I}_n, q_{n,i}(0)=\frac{1}{|\mathcal{I}_n|} $. In the
following sub-section we show that in the case of two players, both
having two choices, $q(0)$ is in the basin of attraction of the global
maximum. Then, in Subsections~\ref{sec:3play} and \ref{sec:3choices},
we give counter examples to show that the result does not extend to
the general case of more than two players or more than two choices.

\subsubsection{Case of two players and two choices}

\begin{proposition}\label{prop:2*2}
  In a two players, two actions allocation game with repercussion
  utilities, the initial point of the algorithm is in the basin of
  attraction of the global maximum.
\end{proposition}

\begin{proof}
  Both players $1$ and $2$ can either take action $a$ or $b$. We
  denote by $x$ the probability for player $1$ to choose $a$, and by
  $y$ the probability for $2$ to choose $a$. We denote by
  $K=(k_{i,j})_{i,j\in \{0,1 \}}$ the matrix such that $k_{i,j}\bydef
  F(i,j)$, where $F(x,y)$ is the potential function\footnote{Actually,
    here, the derivative of the potential is equal to the projection
    of the expected payoffs on the set $\Delta_n$.}. Then, the
  dynamics ~(\ref{dynamics}) can be rewritten:
\begin{equation}
\label{eq:deriv}
\left\{
\begin{array}{l @{\;=\;} l}
  \displaystyle \frac{dx}{dt} & x(1-x)(k_{0,1}-k_{0,0}+Ky) \\[0.8em]
  \displaystyle \frac{dy}{dt} & y(1-y)(k_{1,0}-k_{0,0}+Kx),
\end{array}
\right.
\end{equation}
where $K=k_{1,1}+k_{0,0}-k_{0,1}-k_{1,0}$. Note that in a two-player
two-action game, there are at most two local maxima. Suppose that in
the considered game, there are two local maxima. They are necessarily
attained either at points $(0,0)$ and $(1,1)$ or at points $(0,1)$ and
$(1,0)$. Without loss of generality, we can assume the former case.
Hence, $k_{0,0}$ and $k_{1,1}$ are local maxima, and $k_{1,1} >
k_{0,0}+\gamma $, where $\gamma >0$.

We now define set $E$ and function $V$ as follows:
\begin{equation*}
\begin{array}{lll}
V(x,y)&=&|1-x|+|1-y|, \\
E&=&\{(x,y):x+y>1,\, 0<x,\, y<1 \}.
\end{array}
\end{equation*}
($V$ is actually the distance of $(x,y)$ to the point $(1,1)$ for the 1-norm.)
We next show that $V$ is a Lyapunov function for the dynamics on the
open set $E$. To prove this, it is sufficient to show that
$$L(x,y) \bydef \frac{\partial
  V}{\partial x}(x,y)\frac{dx}{dt}+\frac{\partial V}{\partial
  y}(x,y)\frac{dy}{dt}<0.$$

First, note that $\forall (x,y) \in E$, $V(x,y)=2-x-y$. Hence, from
Eq.~\ref{eq:deriv},
$$
L(x,y) = - x(1-x)(k_{0,1}-k_{0,0}+Ky)- y(1-y)(k_{1,0}-k_{0,0}+Kx).
$$
Let also be $D$ the open segment $\{(x,y):x+y=1, 0<x,y<1
\}$. Trivially,
\begin{equation}
\label{eq:bord}
\forall (x,y) \in D,\, L(x,y=1-x)=-x(1-x)(k_{1,1}-k_{0,0})<0.
\end{equation}
Let us finally consider the segment 
$$S(x_0)=\{(x,y): x+y \geq 1,x=x_0, 0\leq y \leq 1\}.$$ 
Figure~\ref{fig:cas22} summarizes the different notations introduced.
\begin{figure}[htb]
\centering
\includegraphics[width=0.4\linewidth]{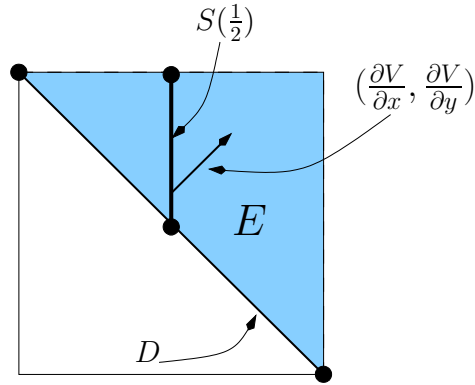}
\caption{Proof of Prop.~\ref{prop:2*2}: Summary of notations}
\label{fig:cas22}
\end{figure}

Since $\displaystyle E \subset \bigcup_{0<x<1}S(x)$, it is sufficient
to show the negativeness of $L$ on $S(x)$ for all $x$. Let us denote
by $L_{x}(y)$ the restriction of $L$ on $S(x)$. From
Eq.~\ref{eq:bord}, we have $L_{x}(1-x)<0$.  Furthermore, $L_x(y)$ is a
quadratic function and its discriminant is
$4(k_{1,0}-k_{0,0})(k_{1,1}-k_{0,1})$, hence is negative. So, for all
$x$, $L_x(y)$ is negative (strictly). Finally, $L$ is negative (strictly) in $E$
and hence non-positive in a neighborhood of $E$.

Therefore, $V$ is a Lyapunov function for the dynamics on a
neighborhood of the open set $E$. More precisely, $V$ is strictly
decreasing on the trajectories of the dynamics starting in the set
$E$, hence they converge to the unique minimum of $V$ which is the
point $(1,1)$. This applies to the initial point $(0.5,0.5)$.
\end{proof}

Figure~\ref{dynamic_2_2} illustrates this result: consider a two
player (numbered $1$ and $2$), two strategy (denoted by $A$ and $B$)
game. Let $x$ (resp. $y$) be the probability for player $1$
(resp. $2$) to take action $A$. While two (local) maxima exist -
namely $(1,1)$ and $(0,0)$ - the surface covered by the basin of
attraction of the global optimum (which is $(1,1)$ in this example) is
greater than those of the other one. A by-product is that the dynamics
starting in point $(0.5,0.5)$ converges to the global optimum.

\begin{figure}[htb]
\begin{center}
  \includegraphics[width=0.5\linewidth, angle=270]{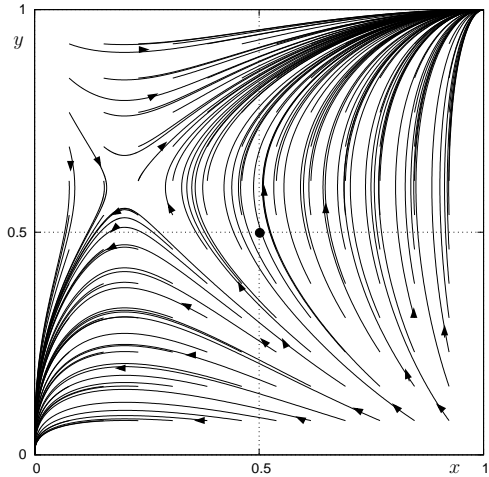}
  \caption{An example with $2$ players with  $2$ choices each. There are $2$
    maxima. The point $(\frac{1}{2},\frac{1}{2})$ is inside the
    attracting basin of the global maximum. }\label{dynamic_2_2}
\end{center}
\end{figure}

Unfortunately, this appealing result
cannot be generalized to more players or more actions, as exemplified
in the following subsections.

\subsubsection{Extension to more than two players}
\label{sec:3play}

\begin{example}
  Let us consider a three player game : $({\cal N},{\cal I},{\cal U})$
  with ${\cal N} = \{1,2,3\}$, ${\cal I} = \{A,B\}$, and ${\cal
    U}=(u_n(i,j,k))_{n \in \{1,2,3\}, i,j,k \in \{A,B\}}$, where $i$
  (resp. $j$), denotes the choice of player $1$ (resp. $2$). The
  matrix representation of $(u_1,u_2,u_3)$ are given below:
$$
(u_1,u_2,u_3)(i,j,1) = \left(\begin{array}{cc}
    (9,6,4)&(5,5,5)\\
    (5,8,1)&(2,4,4)
 \end{array}\right),
$$ 
$$(u_1,u_2,u_3)(i,j,2) = 	 
\left(\begin{array}{cc}
    (7,2,8)&(5,4,7)\\	
    (6,3,3)&(10,2,8)
  \end{array}\right).
$$
Note that this game has no pure strategies Nash equilibrium and a
single mixed strategies Nash equilibrium, which is
$(x,y,z)=(1/3,5/6,0)$. The corresponding value of the potential
function is $87/6=14.5$.

The repercussion utility matrices are:
$$
(r_1,r_2,r_3)(i,j,1) = \left(\begin{array}{cc}
    (10,9,10)&(6,5,5)\\
    (5,5,6)&(1,1,4)
  \end{array}\right),
$$ 	
$$(r_1,r_2,r_3)(i,j,2) =
\left(\begin{array}{cc}
    (6,4,8)&(5,3,7)\\
    (1,3,4)&(9,11,14) 	 
  \end{array}\right).
$$
This game has two pure Nash equilibria, that are $(x,y,z)=(1,1,1)$ and
$(x,y,z)=(0,0,0)$, corresponding to values of the potential function
that are respectively $29$ and $34$.

Figure~\ref{3_players} shows that the trajectory starting at point
$(\frac{1}{2},\frac{1}{2},\frac{1}{2})$ converges to the local maximum
$(x,y,z)=(1,1,1)$ instead of the global maximum
$(x,y,z)=(0,0,0)$. Note that the performance of the local maximum is
way ahead that of the Nash Equilibrium in the original game.
\end{example}

\begin{figure}[htb]
\centering
  \includegraphics[width=0.5\linewidth, angle=270]{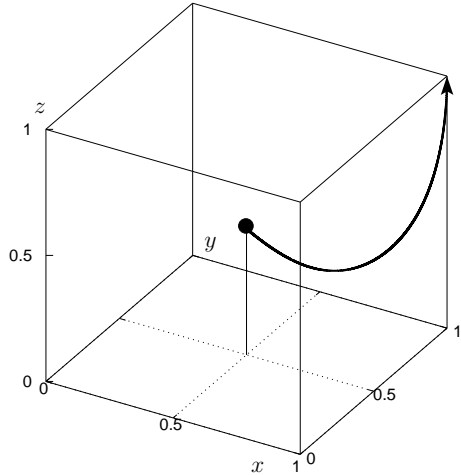}
  \caption{Example with $3$ players, with $2$ choices each. The figure
    represents the dynamic trajectory starting from the point
    $(x,y,z)(0)=(\frac{1}{2},\frac{1}{2},\frac{1}{2})$, with $x$
    (resp. $y$, $z$) the probability for player $1$ (resp. $2$, $3$)
    to adopt action $A$. 
    The dynamics converges to the point $(1,1,1)$ whereas the global
    maximum is $(0,0,0)$.\label{3_players}}
\end{figure}

\subsubsection{Extension to more than two choices}
\label{sec:3choices}

\begin{example}
  Let us now consider the two player game $({\cal N},{\cal I},{\cal
    U})$ with ${\cal N} = \{1,2\}$, ${\cal I} = \{A,B,C\}$, ${\cal U}
  = (u_n(i,j))_{n\in\{1,2\}, i \in \{A,B\}, j \in \{A,B,C\}}$. (Note
  that in this example, only the second player has three possible
  choices).

The payoff matrix is:
$$
(u_1,u_2)(i,j)=
\left(
\begin{array}{ccc}
(6,3)&(-3,11)&(-3,10)\\
(0,2)&(-1, 1)&(0,10)\\
\end{array}
\right).
$$

The companion game is:
$$
(r_1,r_2)(i,j)=
\left(
\begin{array}{ccc}
(7,12)&( -3,11)&(-3,10)\\
(0, 2)&(-11, 0)&( 0,10)\\
\end{array}
\right).
$$

The original game has one single pure Nash equilibria which is $(B,C)$
resulting in the value $10$ for the potential function and no mixed
strategies equilibria exists.

The companion game has two pure Nash equilibria that are $(A,A)$ and
$(B,C)$, corresponding to values of the potential function of $9$ and
$10$ respectively.

Denote $x$ the probability for player $1$ to choose action $A$ and
$y_1$ (resp. $y_2$) the probability for player $2$ to choose action
$A$ (resp. $B$). Then, the global maximum of the potential function is
$10$, and is attained when $x=y_1=y_2=0$. Figure ~\ref{3_choices}
shows that the trajectory starting at point
$(\frac{1}{2},\frac{1}{3},\frac{1}{3})$ converges to the local maximum
$(1,1,0)$, corresponding to Nash equilibrium $(A,A)$ of the companion
game, which is inefficient. Interestingly in this example, the unique
Nash equilibrium of the original game corresponds to the global
maximum of the game.
\end{example}

\begin{figure}[h]
\centering
  \includegraphics[angle=270,width=0.5\linewidth]{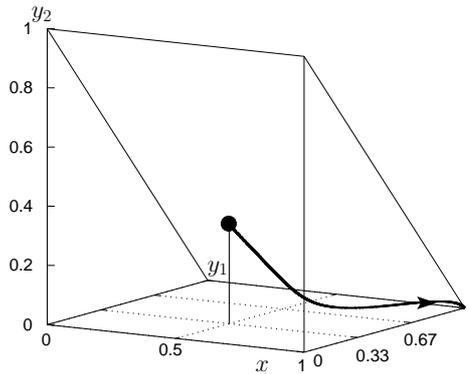}%
    \caption{Example with $2$ players. The first one has $2$ choices and
    the second one has $3$ choices. Here we display the
    $3$-dimensional plot of $y_1$ vs $x$ and $y_2$ vs $x$. The
    dynamics starting in $(1/2,1/3,1/3)$ converges to the point
    $(1,1,0)$ whereas the global maximum is
    $(0,0,0)$.\label{3_choices}}
\end{figure}

\section{Numerical study}
\label{sec:validation}

This section is devoted to implementation issues and shows the
numerical tests that were performed so as to study several possible
practical heuristics based on the algorithm.

First, notice that in the algorithm, users only need to know the
repercussion utility on their current cell to compute their new
strategy vector. Also, each base station only needs to know her own
load to compute the repercussion utilities, hence allowing for a fully
distributed algorithm.

During the execution of the algorithm, at each time slot (typically,
frames are sent every 40 ms for video transmission), each user
executes the algorithm independently, updates her probability vector,
makes a choice according to her strategy and sends a packet to the
corresponding base station.  Meanwhile, each base station measures the
throughputs of all mobiles connected to it and computes the
corresponding repercussion utilities. Then, it sends to every user
their individual repercussion utility.

Once a user reaches a pure strategy, she informs all the cells she has
access to. Each cell waits for all users connected to her to converge
before asking them to monitor their repercussion utility. From then
on, any variation of the load is due to an arrival or departure in the
cell. Hence, upon detection of a change of her repercussion utility,
each user reruns the algorithm, starting with a new probability
vector.

In the previous theoretical sections, convergence of the algorithm
have been shown when the step size $\epsilon$ tends to $0$. Here, we
present several simple heuristics with different step size computation
methods. While the convergence step should be small enough to ensure
convergence, larger values are preferable to decrease the algorithm
runtime.
Hence, appropriate trade-offs need to be examined.

In the first subsection, we present the different heuristics
(Subsection~\ref{sec:heuristiques}). We then present the scenario to
be simulated (in terms of number of users and network topology)
(Subsection~\ref{sec:topologie}). To perform the tests, realistic
throughputs need to be chosen for different combinations of loads,
i.e. values of $u(\ell^i)$ for each possible load $\ell^i$. We provide
such values in Subsection~\ref{sec:valeursNum}.  We then compare the
results obtained by the different heuristics, in terms of efficiency
(the quality of the solution) and convergence speed
(Subsection~\ref{sec:comp}). We briefly comment in
Subsection~\ref{sec:fairness} on the impact of fairness on the
resulting association. Finally, in Subsection~\ref{sec:further}, given
the best heuristic, we provide experimental results about: the
scalability of the algorithm on the system size, the adaptation to
arrival or departure of a mobile, the comparison with other policies,
and the adaptation to different kind of traffic.

\subsection{The Different Heuristics for the Steps}
\label{sec:heuristiques}

Each heuristic actually consists of two parts:
\paragraph*{\it A stopping test} As time increases, the probabilities
of choosing each action tends either to $0$ or $1$. So as to speed up
convergence, we consider thresholds $\delta_{m}$ and $\delta_{M}$
such that:
$$
\forall n \in \mathcal{N},  \forall i \in \mathcal{I}_n, 
\left\{
\begin{array}{@{\,}ll}
  q_{n,i}(t+1) \Gets 0 & \text{if } q_{n,i}(t)< \delta_{m}\\
  q_{n,i}(t+1) \Gets 1 & \text{if } q_{n,i}(t)> 1-\delta_{M}.\\
\end{array}
\right.
$$ 
When one of this operation is done, the strategies are normalized to
remain in the strategy set $\Delta$, and to preserve the condition
$\displaystyle \sum_{i \in \mathcal{I}_n}q_{n,i}=1$. In the tests, we
fix $\delta_{m} = 0.05$ and $\delta_{M} = 0.3$.
\paragraph*{\it A step size computation}: different schemes to compute
$\epsilon_n(t)$ are considered.

\subsubsection{Constant Step Size (CSS)}
In this heuristic, the step size is predefined and constant throughout
time: $\forall n \in \mathcal{N}, \forall t, \epsilon_n(t) =
\epsilon$.  For low values ($\text{CSS}_L$), typically $\epsilon=
0.01$, the algorithm converges in almost all cases to the optimal
solution, but at the cost of a high number of iterations. For high
values ($\text{CSS}_H$), typically $\epsilon = 1$, the convergence and
the optimality are not guaranteed anymore. Intermediate values
($\text{CSS}_M$), typically $\epsilon = 0.1$, are possible
compromises.

\subsubsection{Constant Update Size (CUS)}
At each time epoch, each user computes the maximum step size so that
the change of probabilities for all choices, is bounded by a
predefined value $ \Gamma$ (fixed to 0.1 in the experiments):
$$
\forall n \in \mathcal{N}, \forall i \in \mathcal{I}_n, \quad
\text{abs}\left( q_{n,i} (t+1) - q_{n,i} (t) \right) \leq \Gamma.
$$ 
By bounding the update of every user, this scheme yields smooth
changes in the strategy vectors and hence can be expected to follow
the behavior of the differential equations.

\subsubsection{Decreasing Step Size (DSS)}
The underlying idea of this scheme is to use a few iterations with
large steps before using some smaller step sizes.  Indeed, a big step
size lets actions associated to large repercussion utilities to
quickly get high probabilities of occurrence.  Since the algorithm
converges to a Nash Equilibrium regardless of the initial conditions,
using a few large steps amounts in changing the initial conditions so
as to get close to extrema points, and hence to possible pure
strategies Nash Equilibria.  Then, the following iterations with
smaller step sizes correspond to a good approximation of the $CSS_L$
algorithm. These steps confirm (or infer) the fact that the extremal
point closer to the one obtained after the first iterations is (or
not) a Nash Equilibrium.

We consider two variants of the decreasing step size mechanism. The
first one is a cyclic decreasing step size (DSSSA) (in the
experiments, $\epsilon = 3/ (t \mbox { mod } 10)$). During each cycle a
Nash equilibrium candidate is tested.  This is inspired from simulated
annealing approaches.

The second variant (DSSCSS) is a decreasing step size phase followed
by a constant large step size (in the experiments, $\epsilon = 4/ t $
if $t < 120$ and $\epsilon = 4 $ otherwise).  The underlying idea is
that the first phase would stabilize a certain number of users.  Then,
a large step size should improve the convergence speed of the others
to their respective preferable choices.

\subsection{System Scenario}
\label{sec:topologie}

We consider a simple scenario of an operator providing subscribers
with a service available either through a large WiMAX cell or a series
of WiFi hot spots.

For each simulation, a topology is chosen randomly, according to $3$
parameters (the number of users, the number of WiFi hot spots and the
number $I_n$ of possible choices for each user). More precisely, for each
user:
\begin{itemize}
\item The first choice is the WiMAX cell and one of the $8$ possible
  zones (as defined in Section~\ref{sec:valeursNum}), picked at random
  (uniformly).
\item All other $I_n-1$ choices are one of the Wifi cells, picked up
  according to a uniform law.  As explained in
  Section~\ref{sec:valeursNum}, we consider that all mobiles in a
  common Wifi cell receive the same throughput.
\end{itemize}

The strategy vector is initialized with equal probabilities: $\forall
n \in \mathcal{N}, \forall i \in \mathcal{I}_n, q_{n,i} (0) = 1/I_n$.

\subsection{Throughput of TCP sessions in WLAN and WiMAX}
\label{sec:valeursNum}

Computing the throughput experienced by a packet in a wireless
environment is extremely hard due to the complexity of the physical
system (as opposed to wired system, where the physical medium is
separated from the outside world, and hence has reliable properties,
the wireless link quality changes at every instant, due to the
environment: air quality, buildings and physical obstacles,
etc). Therefore, actual closed formula available in the literature
were obtained using strong assumptions on the outside world and do not
refer to throughput of a single packet but of means of flows. Indeed,
as the number of packets in any connection is large, the flow is
usually approximated as a fluid.

In addition, the useful throughput of a connection, also called
\emph{goodput} depends on the network protocol. Roughly speaking, two
main elements have strong impact on the achieved goodput: first is the
physical system, which depends on the technology in terms both of
maximum capacity and multiplexing technology, second is the transport
protocol.  In this simulation study, we consider the case of TCP flows
for which good throughput approximations are available in the
literature. Yet, the use of UDP flows, or a mixture of TCP and UDP
flows do not impact the performance of the algorithm. (Note that
allowing users to use either TCP or UDP protocol for their
transmission amounts, in the algorithm, to considering an additional
zone in the network cell.)

\paragraph*{Equations of throughput in WiFi cells}
Based on \cite{miorandi06}, we consider that the throughput of
connection $i$ is
$$ u_n(\ell^{i}(s)) = \frac{L_{TCP}}{l^i
  × \left( T_{DATA}+T_{ACK}+2T_{TBO}(l^i)+ 2T_{W}(l^i)\right)
}$$ where $l^i=\sum_{n \in \mathcal{N}}\ell^i_n$ is the number of
mobiles connected to network $i$, $L_{TCP}=8000$ bits is the size of a
TCP packet, $T_{ACK}$ is the raw transmission times of TCP ACK
(approximately $1.091$ ms), $T_{DATA}$ the raw transmission times of a
TCP data packet (about $1.785$ ms). Then, $T_W$ and $T_{TBO}$ are the
mean total time lost due to collisions and back-offs
respectively. These depend on the collision probability of each
packet, and hence on the load of the network. This collision
probability can be numerically obtained via a fixed point equation
given in~\cite{miorandi06}.  Figure \ref{fig:cellCapa} displays the
throughput of a WiFi cell, as a function of the load.

\begin{figure}
\begin{center}
\includegraphics[width=0.8\linewidth]{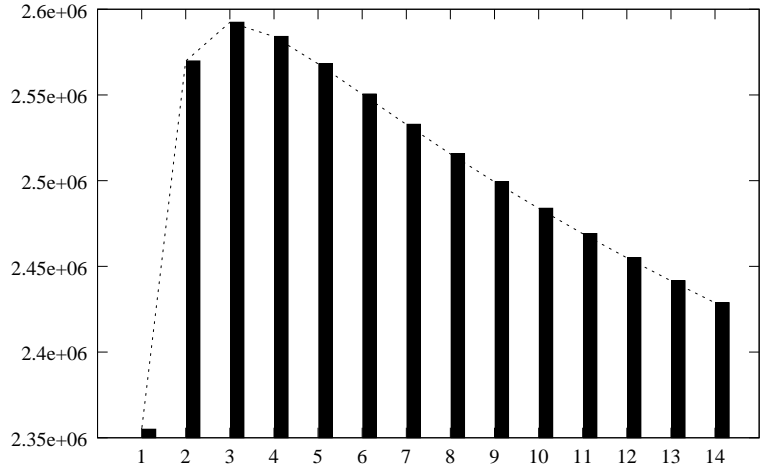}
\caption{Capacity of a WiFi cell as a function of its load (in
  bit/s). The maximum is reached with 3 users.}
\label{fig:cellCapa}
\end{center}
\end{figure}

\paragraph*{WiMAX}
As opposed to WiFi, the WiMAX technology uses OFDMA
multiplexing. Hence, each user receives a certain number of carriers
which are converted into a certain amount of throughput depending on
the chosen modulation and coding scheme, which greatly depends on the
link quality at the receiver side. We consider a fair sharing in terms
of carriers~\cite{tijani}, {\it i.e.}  if $p$ users are present in the
WiMAX cell, each of them will receive $\text{{\it NbSCarriers}}/p$
sub-carriers, similarly to processor sharing. Hence, the goodput
experienced by a user in zone $z$ (corresponding to a coding scheme)
is roughly the fraction $1/p$ of the throughput she would obtain if
she were alone in the cell.

For a single user within the WiMAX cell, we follow experimental values
obtained
in~\cite{perfWiMAX} for IEEE WiMAX 802.16d for its eight zones: \\
[0.6em]
\noindent
\scalebox{.83}{
\begin{tabular}{@{}c|@{}cccc@{}}
Modulation & QAM64 3/4 &  QAM64 2/3 &  QAM16 3/4  &  QAM16 1/2 \\
TCP goodput &  9.58 & 8.88  &  6.80 & 4.50 \\ 
\hline
Modulation & QPSK 3/4  &  QPSK 1/2 & BPSK 3/4 & BPSK 1/2 \\
TCP goodput &3.37  &  2.21 &1.65   & 1.08  \\
\end{tabular}
}

\subsection{Comparisons between Heuristics}
\label{sec:comp}

\begin{figure}
\begin{center}
\includegraphics[width=0.8\linewidth]{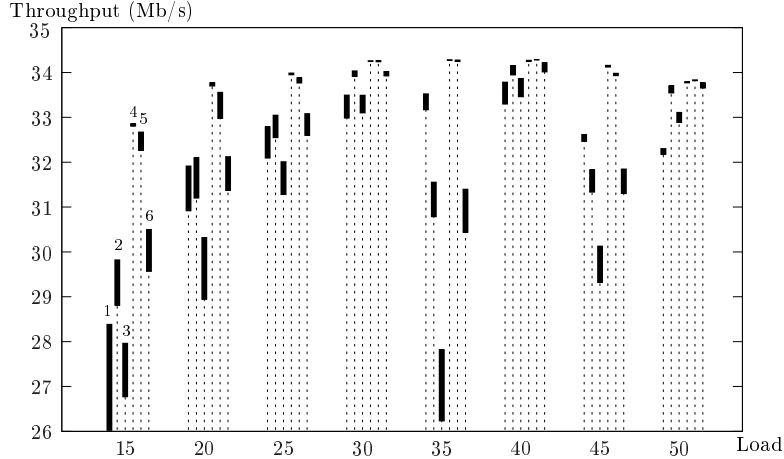}
\caption{Average performance of the  heuristics ($CUS,DSSSA, DSSCSS,
  CSS_L,   CSS_M $ and $CSS_H$ resp.)  with 
  different loads (with 5\% confidence intervals).}
\label{fig:Perf}
\end{center}
\end{figure}

\begin{figure}
\begin{center}
\includegraphics[width=0.8\linewidth]{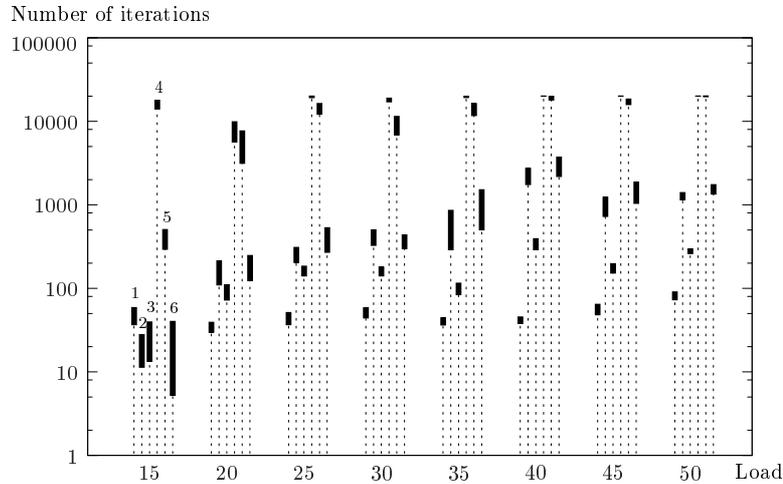}
\caption{Average number of iterations before convergence  of   heuristics ($CUS,DSSSA, DSSCSS,
  CSS_L,   CSS_M$ and $CSS_H$ resp. )  for
  different loads (with 5\% confidence intervals).}
\label{fig:cellIter}
\end{center}
\end{figure}

Figure \ref{fig:Perf} displays the performance (in terms of global
throughput) obtained by the six heuristics ($CUS,DSSSA, DSSCSS, CSS_L,
CSS_M$ and $CSS_H$ resp.) as a function of the total number of users
$N$.  For a given load, all heuristics have been tested on the
same topology to allow a fair comparison.

The small constant step size ($CSS_L$ with $\epsilon = 0.01$), provides the
best performance. It is are even tested optimal for the small values
of $N$, up to 20.

Most heuristics stay within 10 \% of the optimal (except for DSSCSS
whose performance can be poor).  Also note that the total capacity of
the system is less than 36 (10 * 2.6 (WiFi) + 9.58 (WiMAX))
Mbit/s. Thus the best heuristic is always within 5 \% of the optimal.
Finally, it should be noted that the medium constant step size
($CSS_M$) with $\epsilon = 0.1$ is always very close to the best ($CSS_L$)
and that the constant update size ($CUS$) performs better and better
when the number of users grows.

As for the number of iterations, it varies widely between the
different heuristics, even on a logarithmic scale (see Figure
\ref{fig:cellIter}).  The $CUS$ heuristic is a clear winner here (with
an average number of iterations never above 80).  Meanwhile, $CSS_L$
does not always converge within the limit of 20,000 iterations set in
the program.

Under high loads, $CUS$ provides the best compromise with very fast
convergence and reasonable performance.  Under light load, the
constant step size of medium size ($CSS_M$) is also an interesting
choice, for its performance is almost optimal and its number of
iterations remains below 100.

\subsection{Impact on Fairness}\label{sec:fairness}

Consider the following scenario: a set of $20$ users, each having $3$
available choices among 10 cells. The WiMAX cell is numbered 0 and its
8 zones are numbered from 0 to 7.  The set of choices of the users are
$I=$
\begin{equation*}
\begin{array}{lll}
  \{ \{0,1\}, \{ 8\}, \{ 1\} \} & \{ \{0,5\}, \{
  6\}, \{ 4\} \} & \{ \{0,1\}, \{ 6\}, \{ 9\} \}\\ 
  \{ \{0,2\}, \{2\}, \{ 6\} \} & \{ \{0,3\}, \{ 8\}, \{ 9\} \} &
  \{ \{0,6\}, \{4\}, \{ 9\} \}\\
  \{ \{0,7\}, \{ 3\}, \{ 6\} \} & \{ \{0,4\}, \{
  1\}, \{ 2\} \} & \{ \{0,6\}, \{ 6\}, \{ 9\} \}\\
  \{ \{0,5\}, \{3\}, \{ 4\} \} & \{ \{0,6\}, \{ 3\}, \{ 1\} \} & 
  \{ \{0,7\}, \{9\}, \{ 6\} \}\\
  \{ \{0,3\}, \{ 8\}, \{ 1\} \} & \{ \{0,6\}, \{4\}, \{ 7\} \} &
  \{ \{0,6\}, \{ 9\}, \{ 5\} \}\\ 
  \{ \{0,0\}, \{6\}, \{ 5\} \} & \{ \{0,5\}, \{ 4\}, \{
  1\} \} & 
  \{ \{0,6\}, \{
  6\}, \{ 4\} \}\\
  \{ \{0,3\}, \{ 3\}, \{ 4\} \} & \{ \{0,3\}, \{
  8\}, \{ 4\} \}.\\
\end{array}
\end{equation*}

The optimal association scheme, for $\alpha = 0$ (efficient scheme)
and $\alpha = 2$ (fair schemes) are respectively:
\begin{eqnarray*}
A_\text{eff} &= \{2,
1, 2, 1, 1, 1, 1, 2, 2, 2, 1, 1, 2, 2, 2, 0, 2, 1, 1, 1\},\\
A_\text{fair} &= \{0, 1, 0, 1, 0, 2, 1, 2, 1, 1,
2, 1, 1, 2, 2, 2, 1, 2, 0, 1\}
\end{eqnarray*}
resulting in throughputs of:
\begin{eqnarray*}
\scalebox{.8}{$ 
\begin{array}{l@{\,}l@{\;}l@{\;}l@{\;}l@{\;}l@{\;}l@{\;}l@{\;}l@{\;}l@{\;}l@{}}
 T_{\text{eff}} \hspace{-0.15em}= &0.824, &1.225, &0.824, &1.225, &1.225, &1.225, &0.824, &1.225,
  &0.824, &1.225,\\ &
  0.824, &0.824, &0.824, &2.245,& 2.246, &9.58, &0.824, &1.225, 
  &0.824, &1.225.
\end{array}
$}\\
\scalebox{.8}{$
\begin{array}{l@{\,}l@{\;}l@{\;}l@{\;}l@{\;}l@{\;}l@{\;}l@{\;}l@{\;}l@{\;}l@{}}
 T_{\text{fair}} \hspace{-0.15em}=&  2.22, &1.225, &2.22, &1.225, &1.125, &1.225, &1.225,
  &1.225, &1.225, &1.225, \\ &2.245, &1.225, &1.225, &2.246,
  &1.225, &1.225, &1.225, &1.225, &1.125, &1.225.
\end{array}$}
\end{eqnarray*}

The efficient scheme achieves a total throughput of $31.29$ Mb/s.  The
fair scheme suffers a degradation of slightly less than 10\%, with a
total throughput of $28.34$ Mb/s. Yet a closer look at the figures
indicates that the efficient scheme leads to high differences between
users (user $1$ only obtains a throughput of 0.8 Mb/s while user $16$
is granted 9.58 Mb/s). As for the fair association scheme, on the
other hand, all users benefit from throughputs higher $1.1$ Mb/s. As
in bandwidth allocation mechanisms in wired systems\cite{equite}, the
parameter $\alpha$ hence allows to finely tune the compromise between
maximum global throughput and fairness between users.

To understand these differences, let us compare the loads between the
associations:
\begin{equation*}\scalebox{.92}{$
  L_\text{eff}^{wifi} = \{3,2,3,2,1,2,1,2,3\}, \quad L_\text{fair}^{wifi}
  = \{1,2,2,2,2,2,1,2,2\}$}.
\end{equation*}
From Fig.~\ref{fig:cellCapa}, one can see that the maximum capacity
for the WiFi cells is obtained for a load of 3 users. Hence, the
efficient scheme tries to obtain as many cells with load 3 as
possible. Meanwhile, the WiMAX capacity is maximal when its users all
belong to zone $0$. Hence, such users are automatically associated to
this cell (in our case there is only one such user, which obtains a
throughput of $9.58$ Mb/s).

On the other hand, the fair scheme tries to find balanced association
schemes. Hence, the loads of the different WiFi cells are close to one
another\footnote{Note that they cannot be strictly equal due to the
  discrete nature of the problem.} (here ranging between $1$ and $2$)
and the WiMAX cell is associated to some users belonging to efficient
zones. Their number is chosen so as to obtain similar performance as
for users remaining in the WiFi cells.

Hence, while purely efficient schemes produce lightly loaded WiMAX
cells (with only the users in zone $0$), the fair scheme leads to more
balanced loads (here, $4$ users in the WiMAX cell and about $2$ users
in each WiFi cell).

\subsection{Further simulations}
\label{sec:further}
While very small constant step sizes provided limit points with near
optimal performance, all heuristics but CUS needed several thousand
steps before convergence for scenarios with more than 10 users and/or
cells.  The number of steps for CUS never topped 100 and its limit
points also proved very good (a few percent of the optimal). All
simulations reported in this subsection use the CUS heuristic.

\subsubsection{Scalability}
Here, we investigate the impact of the number of mobiles and the
number of cells each mobile can connect to on the speed of convergence
(Figures~\ref{fig:cus2},\ref{fig:cus3})). Unlike in the previous
section where the criterion of convergence speed was the number of
iterations of the algorithm, here, we measure the average number of
handovers for a mobile before convergence. It can be argued that this
new measure of convergence is more relevant since handovers are costly
for mobiles. Figures~\ref{fig:cus2},\ref{fig:cus3} show that the mean
number of handovers is smaller than $20$ when mobiles have $2$
choices, and smaller than $25$ when mobiles have $3$ choices, even for
large numbers of mobiles.

\begin{figure}
\begin{center}
\includegraphics[width=0.5\linewidth, angle=270]{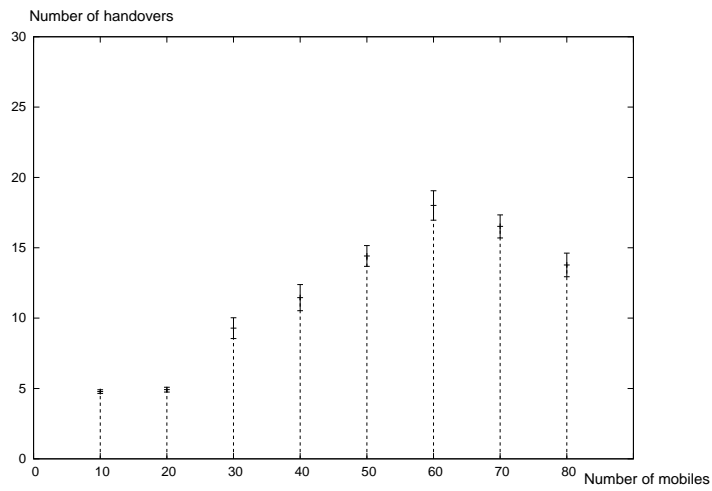}
\caption{Mean number of handovers for a mobile when she has 2
  choices, as a function of the total number of mobiles (full lines
  represent the average measure and the upper and lower 5\% confidence
  interval).}
\label{fig:cus2}
\end{center}
\end{figure}

\begin{figure}
\begin{center}
\includegraphics[width=0.5\linewidth, angle=270]{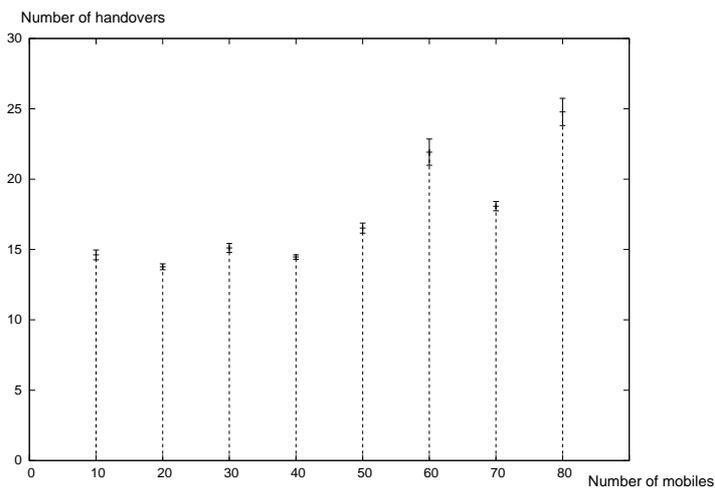}
\caption{Mean number of handovers for a mobile when she has 3
  choices, as a function of the total number of mobiles.}
\label{fig:cus3}
\end{center}
\end{figure}

\subsubsection{Adaptation to Arrivals and Departures}

The association algorithm has to be run at every arrival or departure
of a user in a cell. Here, we simulate the occurrence of such events.
Typical time scales compare nicely: while arrivals or departures of
users in WiMAX or WiFi cells occur every minute or so, the association
algorithm converges in less than a second in most
cases.

\begin{figure}
\begin{center}
\includegraphics[width=0.5\linewidth, angle=270]{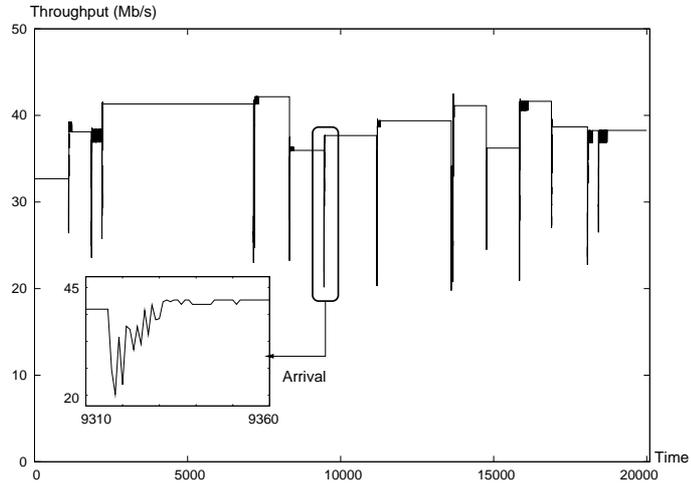}
\caption{Adaptation to arrivals and departures: the heuristic smoothly
  and quickly reconverges after state change.}
\label{dynamic_arrivee_alea}
\end{center}
\end{figure}

\begin{figure}
\begin{center}
  \includegraphics[width=0.5\linewidth,
  angle=270]{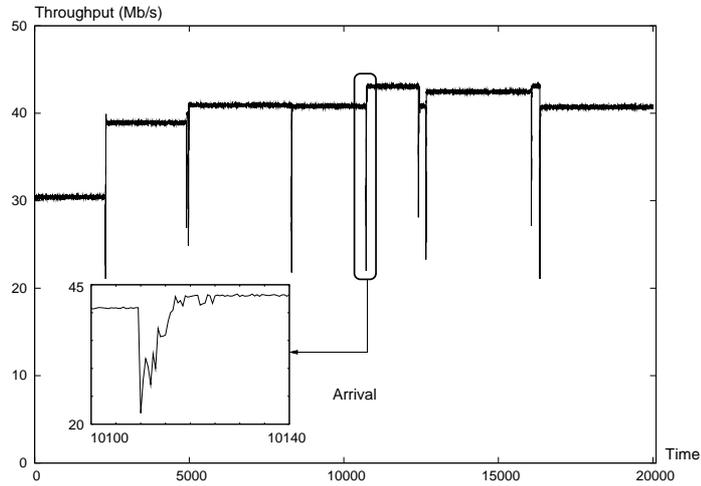}
  \caption{Stability with respect to measurement errors: behavior of the algorithm when the  throughput of all cells
    has a white Gaussian noise with $0.45$ variance.}  \label{dynamic_debit_alea}
\end{center}
\end{figure}

In Figures~\ref{dynamic_arrivee_alea}, \ref{dynamic_debit_alea}, the
arrivals follow a Poisson process. Each incoming mobile has a message
of exponential random size to download. One unit of time corresponds
to the duration of an iteration of the algorithm. In the second
figure, white noise may model perturbations on the cell capacity
(fading) as well as errors on the measures of the real throughput.

\subsubsection{Comparison with Naive or Sub-optimal Methods}

In this section, we compare our algorithm to naive allocation methods
for incoming mobiles.

\paragraph{Fixed Allocation to a WiFi Cell.}
The first naive method for a mobile consists in always connecting to a
WiFi cell if it is possible. It is inspired by the only currently
deployed technology implementing vertical handovers called GAN
(Generic Access Network), also known as Unlicensed Mobile Access
(UMA). Actually, GAN only enables to switch between WLAN and
GSM/UMTS. The capacity of WLAN networks is so much larger than the one
of GSM/UMTS networks that switching to WLAN network whenever possible
is almost always a good choice.  That is why the network selection of
GAN is very basic: the handset gives absolute preference to 802.11
networks over GSM. However, the GAN selection scheme is unlikely to be
efficient in more complex settings, especially when the load of WiFi
cells becomes very large and when WiFi cells compete against WiMAX or
LTE cells whose performance are closer to WiFi than UMTS.
Figure~\ref{fixed2} shows the relative improvement of our algorithm
compared to GAN-like approach.

\begin{figure}
\begin{center}
\includegraphics[width=0.49\linewidth, angle=270]{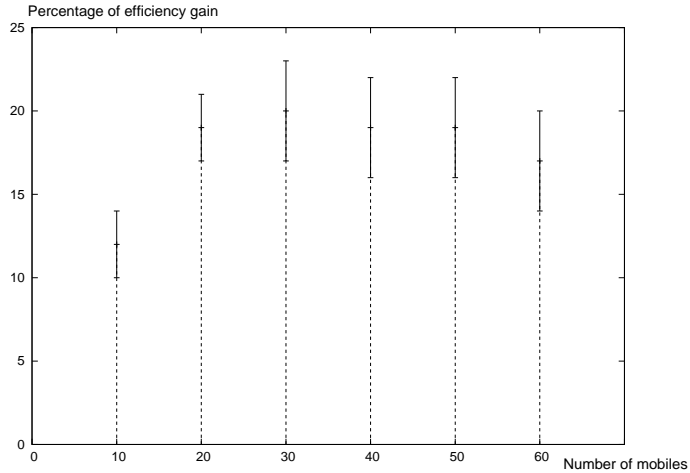}
\caption{Percentage of efficiency gain by using our algorithm in
  comparison to the fixed choice of WiFi cell for each incoming
  mobile. The number of mobiles is variable, but the number of WiFi
  cells is fixed to $15$.\label{fixed2}}
\end{center}
\end{figure}

\paragraph{Allocation to the Best Cell.}
\begin{figure}
\begin{center}
\includegraphics[width=0.5\linewidth, angle=270]{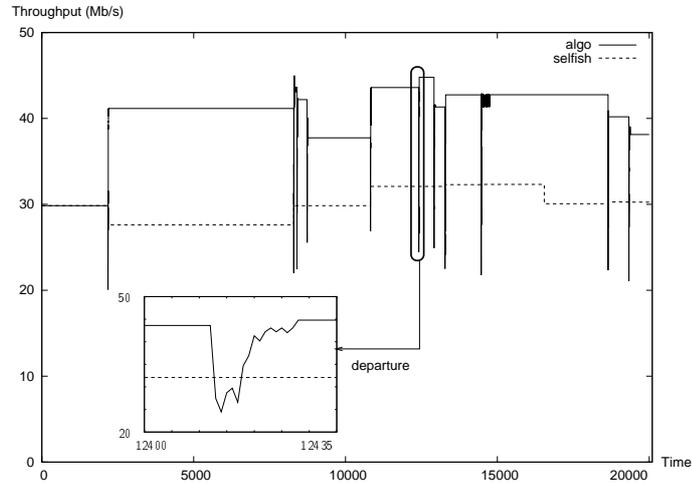}
\caption{Evolution of the global throughput using the association
  algorithm ("algo") and greedy probing ("selfish"). At time $0$, the
  configuration is the same, and the arrival processes of users are
  identical in the two cases.  Since the throughputs for mobiles are
  different in the $2$ schemes, the departure times are different. The
  mean performance in this period of time is $40.1$ for (algo), and
  $29.9$ for (selfish).}\label{compare_choix_best_debit}
\end{center}
\end{figure}
As for this second naive method, an incoming mobile acts selfishly:
she probes all available cells and always connects to the one that
offers the best throughput at connection time and does not change ever
after. Figure (\ref{compare_choix_best_debit}) shows the difference of
the global throughput when we use the both methods of association. We
see that our algorithm achieves a significant better throughput than
the selfish method. This is yet another illustration of the fact that
selfish behaviors lead to a bad use of the resources.

\paragraph{Comparison with the Throughput as Payoff.}
At last, we compare our algorithm when we use the repercussion utility
as payoffs for mobiles (Section~\ref{s-reper}) which ensures the
convergence to an locally optimal point, to the same algorithm when
the payoff is equal to the throughput: $r_n \bydef u_n$ for all
users. See
Figures~\ref{compare_reward_debit_rapport_constant},\ref{compare_reward_debit_rapport_variable_trente}
and \ref{compare_reward_debit_rapport_variable_vingt}. Here the gain
is much lower but both algorithms roughly have the same convergence
time.

\begin{figure}
\begin{center}
\includegraphics[width=0.5\linewidth,
angle=270]{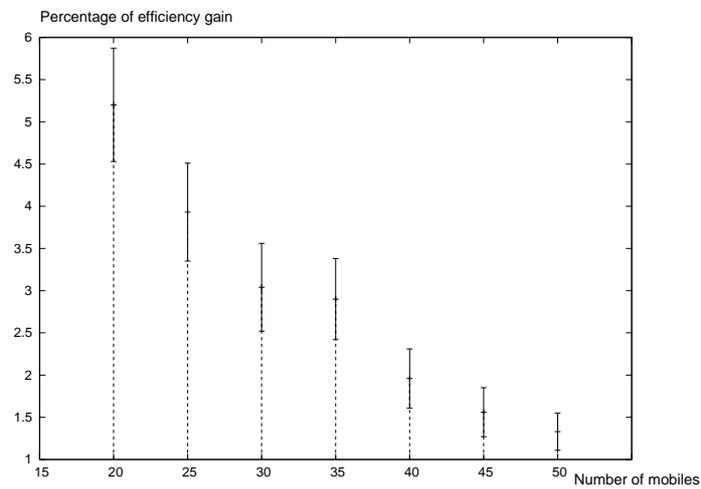}
\caption{Percentage of efficiency gain when using repercussion
  utilities instead of throughputs, when the number of mobiles
  varies. The ratio of the number of WiFi cells divided by the number
  of mobiles is constant and equal to $1$ over
  $5$.}\label{compare_reward_debit_rapport_constant}
\end{center}
\end{figure}

\begin{figure}
\begin{center}
\includegraphics[width=0.5\linewidth,
angle=270]{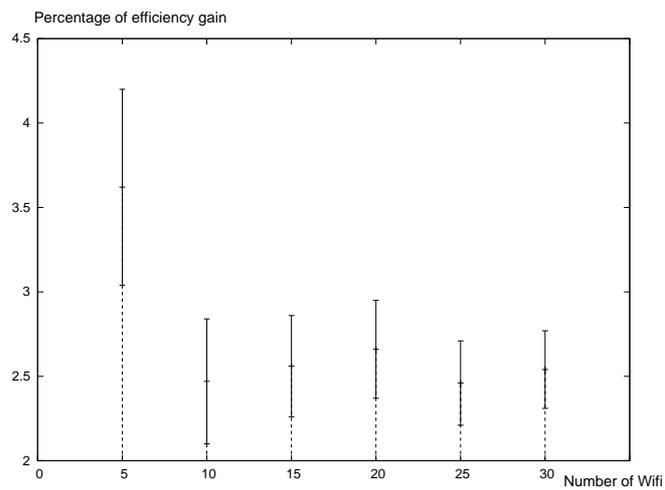}
\caption{Similar to
  Figure~\ref{compare_reward_debit_rapport_constant}, but the number
  of WiFi cells varies and the number of mobiles is constant and equal
  to $30$.}\label{compare_reward_debit_rapport_variable_trente}
\end{center}
\end{figure}

\begin{figure}
\begin{center}
  \includegraphics[width=0.5\linewidth,
  angle=270]{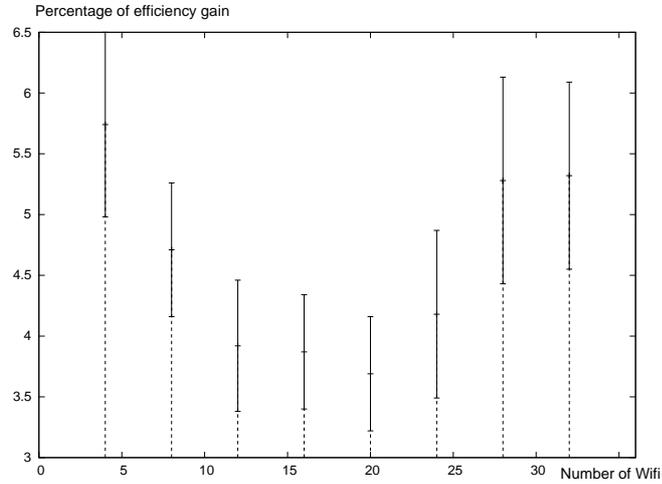}
  \caption{Like
    Figure~\ref{compare_reward_debit_rapport_variable_trente}, but
    with a constant number of mobiles equal to
    $20$.}\label{compare_reward_debit_rapport_variable_vingt}
\end{center}
\end{figure}

\subsubsection{Real-Time  Traffic vs Elastic Traffic}

The question here is to know whether real time traffic can be taken
into account in the algorithm. In fact, for elastic traffic, utility
for users is intimately related to the throughput they receive. For
real time traffic like voice or video transmission, users require a
certain level of throughput.  Hence the idea is to build a different
utility function for these users.

The first idea is to have a null utility if the throughput is under a
certain threshold, and a utility equal to $1$ otherwise. The algorithm
works well with this utility but is long to converge because the
discontinuity causes a bang-bang behavior of the users. This problem
can be avoided by transforming the utility function: under the
threshold the utility is still $0$, and becomes
$1-exp(-u_n(\ell^{s_n}))$ above it. This provides good solutions in
terms of convergence speed as well as a good overall utility. In
Figure~\ref{real_time}, we show the behavior of the time of
convergence of this heuristic when the ratio of real-time traffic
vary. The impact of this ratio on the time of convergence is not
significant.

\begin{figure}
\begin{center}
  \includegraphics[width=0.5\linewidth,
  angle=270]{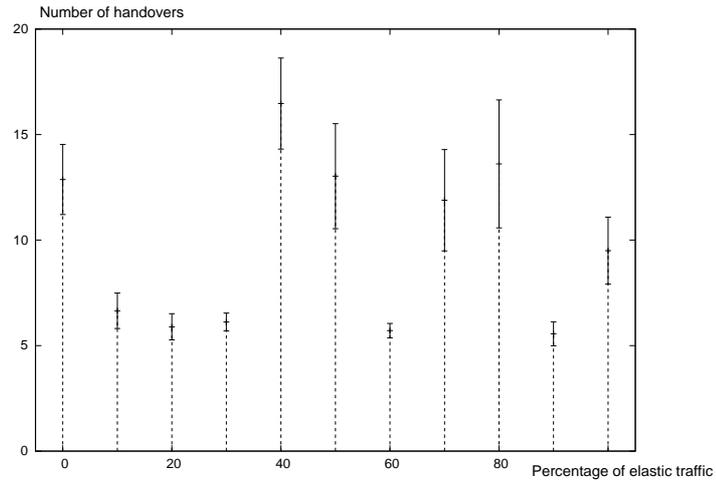}
\caption{Dependency of the time of convergence when  the ratio of
  elastic traffic varies. The number of mobiles is 30.}\label{real_time}
\end{center}
\end{figure}

\subsubsection{A Dynamic Scenario: between Mice and Elephants}

Here, we consider that the global traffic is shared by two kinds of
traffic called {\it mice} and {\it elephants}. The mouse traffic
corresponds to short lived connections ($< 1$ second) and the elephant
traffic to long connections (up to one minute).  There are relatively
few elephants and a large number of mice ($90 \%$), but globally, the
ratio of elephant traffic represents approximately $85 \%$ of the
global traffic. Whereas our algorithm is well adapted to elephant
traffic, since the time of convergence is negligible with respect to
the duration of the connection, it is not the case for mice
traffic. In Figures~\ref{mice} and \ref{mice_algo}, we compare two
scenarios, when both mice and elephants use the algorithm and when
only elephants do so (while mice always connect to one WiFi cell). The
second method reduces the number of handovers and preserves the
overall throughput (even giving a small gain) as seen in the
Figures~\ref{mice} and \ref{mice_algo}.

At last, Figure~\ref{mice_ratio} shows the performance gain when we
apply the algorithm for mice and elephants in comparison with applying
it only to the elephants. It points out the fact that both methods
have a similar efficiency, but the second ensures a low rate of
handovers. It is interesting to notice that this is independent of the
ratio of mice traffic. That means that the loss of throughput due to
the algorithm (which is important when the percentage of mice is
high), is balanced by the loss of optimality of the second method.

\begin{figure}
  \begin{center}
    \includegraphics[width=0.5\linewidth,
    angle=270]{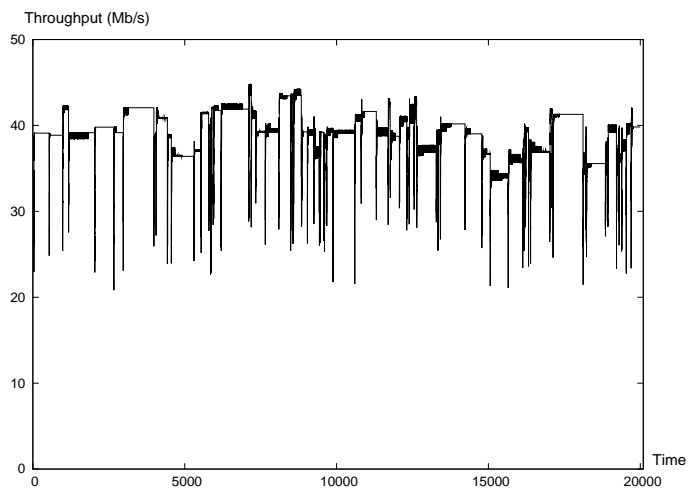}
    \caption{Traffic made of 30 initial users with $90\%$
      mice. Average packet size for elephants is 20 times the average
      packet size for mice.  The figure shows the total throughput
      when all users apply the algorithm. The average total throughput
      is $39.05 Mb/s$.}\label{mice}
\end{center}
\end{figure}

\begin{figure}
  \begin{center}
    \includegraphics[width=0.5\linewidth,
    angle=270]{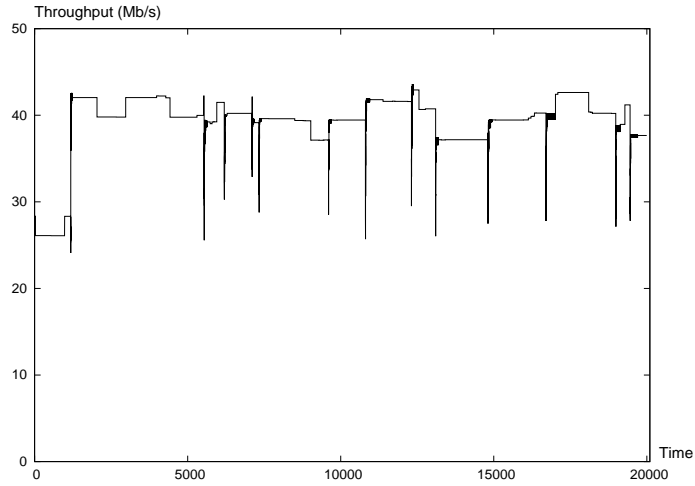}
    \caption{Same configuration and arrival process as in
      Figure~\ref{mice}. In this figure, mice are directly allocated
      to the WiFi cell without applying the algorithm. The mean
      throughput is $39.19 Mb/s$.}\label{mice_algo}
  \end{center}
\end{figure}

\begin{figure}
  \begin{center}
    \includegraphics[width=0.5\linewidth,
    angle=270]{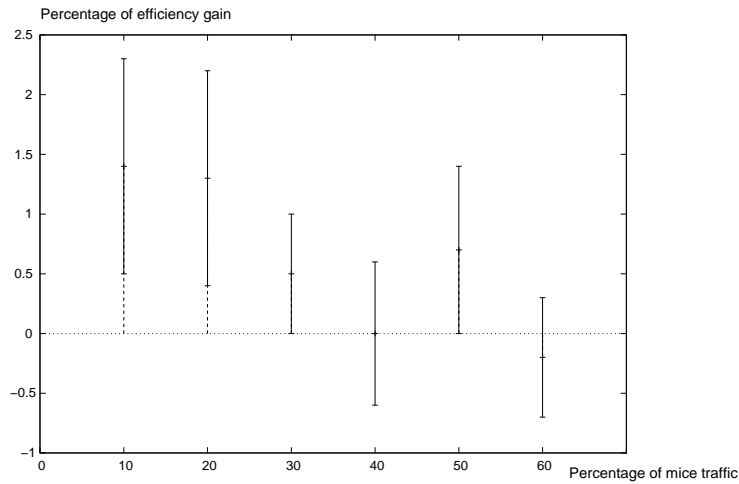}
    \caption{Percentage of gain by running the algorithm for mice and
      elephants instead of running it only for elephants as a function
      of the  percentage of mice traffic ( the global traffic average
      remains  constant).}\label{mice_ratio}
  \end{center}
\end{figure}

\section{Conclusion and Future Works}
In this paper, we have designed a distributed algorithm that selects
an efficient (in terms of fairness or global throughput) network
association in heterogeneous wireless networks. Simulations show that
this method is relevant, in comparison with naive method. This opens
the way to several interesting future works, such as the
implementation of such methods in modern mobile devices in
collaboration with Alcatel-Lucent.

\section{Acknowledgement}
We wish to thank Laurent Thomas, Sabine Randriamasy and Erick Bizouarn
from Alcatel-Lucent, who helped us with the design of several
scenarios using realistic parameters.

\bibliographystyle{IEEEtran}
\bibliography{handover}

\end{document}